%% file: HancovaEBLUPNE2019.tex
\RequirePackage{fix-cm}
\documentclass{svjour3}                     
\smartqed  
\input{OwnSymbolsCommands.tex} 

\begin{document}

\title{Estimating variances in time series  \\
	linear regression models using empirical BLUPs \\
	and convex optimization 
\thanks{This work was supported by the Slovak Research and Development 
Agency under the contract No. APVV-17-0568, the Scientific Grant Agency 
of the Slovak Republic (VEGA), VEGA grant No.1/0311/18 and the Internal Research Grant System of 
Faculty of Science, P.~J. \v{S}af\'arik University in Ko\v{s}ice (VVGS PF UPJ\v{S}) --- project 
VVGS-PF-2018-792.}
}
%

\subtitle{}

\titlerunning{Estimating variances in FDSLRMs via EBLUPs and convex optimization}        

\author{Martina Han\v{c}ov\'a  \and
        Gabriela Voz\'arikov\'{a}\and
        Andrej Gajdo\v s \and
        Jozef Han\v c
}


\institute{M. Han\v{c}ov\'{a} \and G. Voz\'arikov\'a \and A. 
Gajdo\v{s} 
\at
	Institute of Mathematics, Pavol Jozef \v{S}af\'{a}rik University, 
	Ko\v{s}ice, Slovakia \\
	Tel.: +421-55-2342213, Fax: +421-55-6222124\\
	\email{martina.hancova@upjs.sk}           
	\and
	J. Han\v{c}\at
	Institute of Physics,  Pavol Jozef \v{S}af\'{a}rik University, 
	Ko\v{s}ice, Slovakia
}

\date{Received: date / Accepted: date}

\maketitle

\vspace{-12pt}
\begin{abstract}
\leavevmode\\
We propose a two-stage estimation method of variance components in time series 
models known as FDSLRMs, whose observations can be described by a linear mixed model (LMM). 
We based estimating variances, fundamental quantities in a time series forecasting 
approach called kriging, on the empirical (plug-in) best linear unbiased predictions of 
unobservable random components in FDSLRM. The method, providing invariant non-negative quadratic 
estimators, can be used for any absolutely continuous probability distribution of time series data. 
As a result of applying the convex optimization and the LMM methodology, we resolved two problems 
--- theoretical existence and equivalence between least squares estimators, non-negative (M)DOOLSE, 
and maximum likelihood estimators, (RE)MLE, as possible starting points of our method and a 
practical lack of computational implementation for FDSLRM. 
As for computing (RE)MLE in the case of $ n $ observed time series values, we also discovered a new 
algorithm of order $\bigO(n)$, which at the default precision is $10^7$ times more accurate and 
$n^2$ times faster than the best current Python(or R)-based computational packages, namely 
CVXPY, CVXR, nlme, sommer and mixed. We illustrate our results on three real data sets --- 
electricity consumption, tourism and cyber security --- which are easily available, reproducible, 
sharable and modifiable in the form of interactive Jupyter notebooks.


%
%
%

\end{abstract}

\Section{Introduction}
\label{sec:intro}

The need to obtain sufficiently accurate predictions for facilitating and improving decision making 
becomes an integral part not only of science, industry or economy but also in many other human 
activities. Our recent article \cite{gajdos_kriging_2017} summarizes the guiding methodology and 
corresponding references dealing with \textit{kriging} for time series econometric forecasting as 
one of the advanced alternative approaches to the most popular Box-Jenkins methodology 
(\cite{box_time_2015}, \cite{shumway_time_2017}).

The key idea of the time series kriging  is to model the given time series data in an appropriate 
general class of linear regression models (LRMs) \label{abv-LRM} and subsequently finding 
\textit{the best linear 
unbiased predictor} \label{abv-BLUP} --- the BLUP which minimizes the mean squared error (MSE) 
\label{abv-MSE} of prediction among all linear unbiased predictors, see e.g. 
\cite{stulajter_predictions_2002}, \cite{kedem_regression_2005}, \cite{christensen_plane_2011}, 
\cite{brockwell_introduction_2016}.
 
In the frame of kriging, we investigate theoretical features and econometric applications of a 
class of time series models called \textit{finite discrete spectrum linear regression models} or 
shortly FDSLRMs. The FDSLRM class was introduced in 2002-2003 by \v{S}tulajter 
(\cite{stulajter_predictions_2002}, \cite{stulajter_mse_2003}) as a direct extension of classical 
(ordinary) regression models (see e.g. \cite{christensen_plane_2011}, 
\cite{hyndman_forecasting:_2018}).

The FDSLRM has mean values (trend) given by linear regression and random components (error terms) 
are represented as a sum of a linear combination of uncorrelated zero-mean random variables and 
white noise which can be interpreted in terms of the finite discrete spectrum 
\cite{priestley_spectral_2004}. 

Formally the FDSLRM can 
be presented as 
\begin{equation} \label{def:FDSLRM}
	X(t)= \Sum{i}{k}\beta_i f_i(t)+ \Sum{j}{l} Y_j v_j(t)+w(t); \, t \in \Ts,  
\end{equation}
\noindent where 
\begin{description}
	\item[] $\Ts$ representing the time domain is a countable subset of the	real line $\Rv{}$,
	\item[] $k$ and $l$ are some fixed non-negative integers, i.e. $k,l \in \Nz $, 
	\vertspace
	\item[] $\bv = (\beta_1, \beta_2, \ldots, \beta_k)' \in \Rv{k}$ is a vector of regression 
	parameters,
	\item[] $\Yv = (Y_1,Y_2,\ldots,Y_l)'$ is an unobservable $l \times 1$ random vector with zero 
	mean vector $\E{\Yv}=\Ov{l}$, and with an $l\times l$ diagonal covariance matrix 
	\mbox{$\Cov{\Yv} = \diag{\s_j^2}$}, where  $\s_j^2 \in \Rv{}_+$  are non-negative real numbers,
	\vertspace
	\item[] $f_i(.); \, i=1,2,\ldots,k$ and  $v_j(.); \, j=1,2,\ldots,l$ are real functions defined 
	on $\Rv{}$, 
	\item[] $w(.)$ stands for white noise uncorrelated with $\Yv$ and having a positive dispersion 
	$\Di{w(t)}=\s_0^2 \in \Rv{}_{++}$.
\end{description}

Typically, due to the nature of time series data collection, the most frequently considered time 
domain $\Ts$ is the set of natural numbers $\N=\{1,2, \ldots\}$. FDSLRM variance parameters,  which 
are fundamental quantities in kriging, are commonly described by one vector $\nuv   
\equiv (\nus_0, \nus_1, \ldots, \nus_l)' = (\s_0^2,\s_1^2,\ldots,\s_l^2)'$, an element of the 
parametric space $\Ys = (0, \infty) \times [ 0, 
\infty)^{l}$ or $\Ys = \Rv{}_{++} 
\times \Rv{l}_+$  for short.

The principal goal of our paper is to present an alternative estimation method for 
FDSLRM variances, one of the time series kriging steps \cite{gajdos_kriging_2017}. Our approach 
represents an improved two-stage modification of the previously developed method of natural 
estimators 
(NE) \cite{hancova_natural_2008} which is now based on the idea of empirical (plug-in) BLUPs. We 
will refer to our new method as EBLUP-NE for short. 

The structure of the paper is organized as follows. Section 2 includes several important notes 
about theoretical methods and computational tools for FDSLRM kriging offered by closely 
connected mathematical branches and computational research. This background became the strong basis 
for 
our research.

Section 3 describes the first stage of our EBLUP-NE method. It develops the definition and 
computational form of EBLUP-NE at known variance parameters $\nuv$ using two special matrices known 
in linear algebra as the Schur complement and Gram matrix \cite{zhang_schur_2005}. The method 
requires very minimal distributional assumptions on a finite FDLSRM observation $\Xv = 
(X(1),\ldots,X(n))',$ $n \in \N$. It must have only an absolutely 
continuous probability distribution with respect to some $\sigma$-finite measure. In section 3 we 
also derive the basic statistical properties of EBLUP-NE at known $\nuv$. 

Section 4 is devoted to the second stage of EBLUP-NE method which consists in replacing true,  
in practice unknown variance parameters $\nuv$ by another, appropriate FDSLRM estimates of $\nuv$ 
obtained by maximum likelihood or least squares \cite{gajdos_kriging_2017}. Since the theory and 
computational aspects of these estimation procedures in FDSLRM fitting based on the projection 
theory in Hilbert spaces are more than 10 years old (\cite{stulajter_predictions_2002}, 
\cite{stulajter_estimation_2004}, \cite{hancova_natural_2008}), we revisit and update them in 
the light of recent advances in closely related linear mixed modeling and convex optimization.

In the following fifth section, we illustrate theoretical results of the paper and the performance 
of 
EBLUP-NE on three real time series data sets --- \textit{electricity consumption}, \textit{tourism} 
and \textit{cyber security}. Section 6 presents the conclusions of the paper. For the sake of paper 
readability, we moved very technical details and proofs to the Appendix. In the Appendix we also 
report a list of acronyms and abbreviations used in the paper (tab. \ref{tab:acronyms}). 

Since the paper is aimed at statisticians, analysts, econometricians and data scientists applying 
time 
series and forecasting methods, our notation is standard for time series analysis and prediction 
using linear regression models (\cite{stulajter_predictions_2002}, \cite{kedem_regression_2005}, 
\cite{brockwell_time_2009}). 
Sets  like $\Rv{}, \Rv{}_{+}, \Rv{}_{++}$ are labeled as it is common in convex optimization 
\cite{boyd_convex_2009}. 

\Section{Theoretical and computational bases for FDSLRM kriging}
\label{sec:basis}

Let us recall some important notes about time series FDSLRM kriging and its 
connection to other math branches and current computational technology. 

From a modeling point of view, if we consider any time series model $X(.)$ in the additive form 
$X(t)=m(t)+\eps(t);t \in \Ts$ whose mean value function $m(.)$ is a	real function on $\Rv{}$ 
expressible by a functional series (e.g. Taylor or Fourier series) and error term 	$\eps(.)$ is a 
mean-zero stationary process, then according to the spectral representation theory of time	series 
(\cite{priestley_spectral_2004},\cite{percival_spectral_2009}) $X(.)$ can be approximated 
arbitrarily 
closely by FDSLRM. Therefore, from a practice perspective, FDSLRMs can be potentially applied 
in many practical situations. 

In econometric applications of FDSLRM kriging, we almost always have only one realization of time 
series $X(.)$ and we do not know the mean-value parameters $\bv \in \Rv{k}$ nor the variance 
parameters $\nuv \in \Ys $.  From the view of time series analysis, the problem of estimating 
variances $\nuv$ belongs to the class of problems concerning covariance matrix estimation with one 
realization \cite{wu_covariance_2012}.
	
In addition, any finite FDSLRM observation $\Xv$ satisfies a 
special type of linear mixed model (LMM) of the form
\begin{equation} \label{eq:FDSLRMobservation}
	\begin{gathered}
	\Xv = \F \bv +\V \Yv + \wv \text{ with } \E{\wv}=\Ov{n}, \Cov{\wv} = \s_0^2\In,  \\
	\Cov{\Yv,\wv} = \Om{l}{n}, \Cov{\Xv} \equiv \Sgn = \s_0^2\In+\V\D\V',
	\end{gathered}
\end{equation} 
where design matrices $\F = \{F_{ti}\} =\{f_i(t)\}, \V=\{V_{tj}\} =\{v_j(t)\} \text{ for } t=1, 2, 
\ldots,n$; $i=1,2,\ldots,k$; $j=1,2\ldots,l$ and random vector $\wv =(w(1), w(2), \ldots, w(n))'$ 
is a finite $n$-dimensional white noise observation. 

The fundamental property \eqref{eq:FDSLRMobservation} of FDSLRM allows us 
to apply convenient LMM mathematical techniques for FDSLRM fitting and 
forecasting (\cite{pinheiro_mixed-effects_2009}, 
\cite{christensen_plane_2011}, \cite{witkovsky_estimation_2012},  \cite{demidenko_mixed_2013}, 
\cite{rao_small_2015}, \cite{covarrubias-pazaran_genome-assisted_2016}, 
\cite{singer_graphical_2017}) together with up-to-date LMM software packages --- \textit{nlme}, 
\textit{sommer} written in R (\cite{pinheiro_nlme:_2018}, \cite{covarrubias-pazaran_sommer:_2019}), 
MATLAB function \textit{mixed} \cite{witkovsky_mixed_2018} or SAS package \textit{Proc MIXED} 
\cite{sasinstituteinc._sas/stat_2018}. As the first practical consequence of the LMM theory, to 
have identifiable FDSRLMs under assumptions of multivariate normal distribution \cite[sec.~3.2, thm 
11]{demidenko_mixed_2013}, in the paper, we will assume for the block 
matrix $(\F \;\V)$ and for the number of data $n$ in $\xv$, a realization of $\Xv$ in real 
situations, the following sufficient condition
\begin{equation}
\emph{identifiability}: \text{the full column rank } r(\F\;\V) = k+l \text{ and } n>k+l. 
\end{equation}

On the other hand, the standard maximum likelihood or least squares estimates of $\nuv$ in FDSLRM 
are optimization problems. Our motivation to explore estimating FDSLRM variances in the frame of 
convex optimization lays on the fact that its mathematical tools and efficient, very reliable 
computational interior-point methods became important or fundamental tools in many other branches 
of mathematics \cite{boyd_convex_2009} like the design of experiments, high-dimensional data, 
machine learning or data mining. 

Inspired by essential and recent works on convex optimization (\cite{boyd_convex_2009}, 
\cite{bertsekas_convex_2009}, \cite{koenker_convex_2014}, \cite{cornuejols_optimization_2018}, 
\cite{agrawal_rewriting_2018}) we prove new theoretical relations among existing FDSLRM estimators 
in the extended parametric space $\Ys = \Rv{l+1}_+$ for the orthogonal version of FDSLRM, most used 
in 
real time series applications. Moreover, we show how to apply the latest convex optimization 
packages CVXPY and CVXR based on disciplined convex programming, written in Python 
\cite{diamond_cvxpy:_2016} and R \cite{fu_cvxr:_2019}, for estimating $\nuv$. Finally, we also 
focus on the development of a new fast and accurate computational optimization algorithm for 
estimating variances in FDSLRM. 
 
As for computational technology, our time series calculations are carried out using free open 
source software based on the R statistical language and packages 
(\cite{rdevelopmentcoreteam_r:_2019}, \cite{mcleod_time_2012}), the Scientific Python with the 
SciPy \label{abv-SciPy} ecosystem (\cite{jones_scipy:_2001}, \cite{oliphant_python_2007}) and free 
Python-based
mathematics software SageMath \label{abv-SageMath}(\cite{stein_sage_2019}, \cite{hogben_sage_2013}, 
\cite{zimmermann_computational_2018}), an open source alternative to the well-known commercial 
computer algebra systems Mathematica or Maple.  The simultaneous use of R, SciPy and SageMath 
provides us valuable cross-checking of our computational results.

All our computational algorithms and results are easily readable, sharable, reproducible,  and 
modifiable 
thanks to open-source Jupyter technology \cite{kluyver_jupyter_2016}. In particular, we present our 
results in the form of Jupyter notebooks, dynamic HTLM documents integrating prose, code and 
results similarly as Mathematica notebooks, stored in our free available GitHub repository 
(\cite{gajdos_fdslrm_2019}, \cite{adar_coding_2014}). 

The Jupyter notebooks, detailed records of our computing with explaining narratives, can be seen or 
studied as static HTML pages via \textit{Jupyter nbviewer} (\url{https://nbviewer.jupyter.org/}, 
\cite{kluyver_jupyter_2016}) or 
interactively as live HTML documents 
using \textit{Binder}  (\url{https://mybinder.org/}, \cite{projectjupyter_binder_2018}) where the 
code is executable. The way and presentation of our computing with real data are inspired by works 
\cite{brieulle_computing_2019}, \cite{weiss_scientific_2017}.

\Section{The BLUP estimation method for $\nuv$}

\subsection{Definition and computational form of estimators}
\label{sec:compform}

In the paper \cite{hancova_natural_2008}, we proposed the method of always non-negative natural 
estimators (NE) \label{abv-NE} of FDSLRM variances $\nuv$. The main idea behind NE came from the 
fact that $\s_j^2 
= \Cov{Y_j} = \E{Y_j^2}; j = 1, \ldots, l$. Therefore, if random vector $\Yv$ in \eqref{def:FDSLRM} 
was known, the natural estimate of $\s_j^2$ would be just $Y_j^2$. This initial consideration was 
identical with Rao's estimates known as MINQUE (\cite[sec.~5.1]{rao_estimation_1988}, 
\cite[sec.~12.7]{christensen_plane_2011}).

To get sufficiently simple, explicit analytic expressions available for further theoretical 
study, we predicted the random vector $\Yv$ by the ordinary least squares method leading to the 
following linear predictor of $\Yv$ based on $\Xv$
$$\breve{\Yv} = \W^{-1}\V'\Mf\Xv; \; \W =\V'\Mf\V \in \Rm{l}{l}.$$
Matrix $\Mf = (\In-\F(\F'\F)^{-1}\F') \in \Rm{n}{n}$ represents (\cite{hancova_natural_2008}) the 
orthogonal projector onto the orthogonal complement of the column space of $\F$ and matrix 
$\W\succ0$ ($\succ0$ means positive definiteness)
is known in linear algebra and statistics \cite[chap. 6]{zhang_schur_2005} as the Schur complement 
of $\F'\F$ in the block matrix $\smallblockmatrix{\F'\F}{\F'\V}{\V'\F}{\V'\V}$.

However, from perspective of the general theory of the  best linear unbiased prediction in LMM
(see e.g. \cite[chap. 12]{christensen_plane_2011}), we could 
consider natural estimators of $\nuv$ in FDSLRM based on the BLUP 
$\blupYn$ of $\Yv$. Then intuitively, it seems reasonable to assume that the BLUP as the 
unbiased linear predictor with minimal MSE could lead to better estimators of $\nuv$. This 
intuition will be confirmed by our following derivations. 

The previous heuristic consideration motivates the following definition. Let us consider FDSLRM 
\eqref{def:FDSLRM} and its observation \eqref{eq:FDSLRMobservation}. Then 
estimators in the form
\begin{equation} \label{def:BLUPNE}
	\besX{j}=(\blupYn)_j^2; \:\blupYn \textrm{  is the BLUP of  } \Yv \textrm{ based on } \Xv 
\end{equation}
are called \textit{natural estimators of $\nuv$ based on the BLUP of $\Yv$} or BLUP-NE for short.

\begin{remark}
\leavevmode\\
\small
Following the analogy with original NE in \cite{hancova_natural_2008}, for the first 
component $\nus_0 = \s_0^2$ of $\nuv$ we could take as a BLUP based natural estimator the sum of 
squares of white noise residuals $\rwn=\Xv-\F\bluebn-\V\blupYn$ divided by the number of degrees of 
freedom in FDSLRM
\begin{equation*}
\besX{0} = \frac{1}{n-k-l} [\Xv-\F\bluebn-\V\blupYn]'[\Xv-\F\bluebn-\V\blupYn],
\end{equation*}
where $\bluebn$ is the best linear unbiased estimator (BLUE) \label{abv-BLUE} of $\bv$. In the LMM 
framework, FDSLRM 
white noise residuals $\rwn$ are also called \textit{conditional residuals} 
\cite{singer_graphical_2017}.

 But now our intuition fails. The computational form and properties of 
 $\blupYn$ depend on all variance components $\nuv$. That dependency also appears in the proposed 
 estimator $\besX{0}$ which causes its biasedness even in a simpler special case of FDSLRM 
 (orthogonality condition, see \eqref{eq:orthogonality}). Therefore, further we do not pay any 
 special 
 attention to its form and properties. As a natural estimator of $\s_0^2$ we will keep the original 
 unbiased invariant quadratic NE  $\nesX{0}$ of $\s_0^2$, eq. (2.2) in \cite{hancova_natural_2008}.
\end{remark}

The BLUP $\blupYn$ of $\Yv$ together with the BLUE $\bluebn$ of $\bv$ can be obtained from 
celebrated Henderson's mixed model equations, MME for short (for MME in the current framework of 
LMM see e.g. \cite[sec.~12.3]{christensen_plane_2011} or \cite[sec.~2]{witkovsky_estimation_2012}). 
MME have the following form in the case of FDSLRM 
\begin{equation} \label{eq:MME}
	\blockmatrix{\F'\F}{\F'\V}{\V'\F}{\dGV} \cdot
	\begin{pmatrix} \bluebn \\ \blupYn \\ \end{pmatrix} = 
	\begin{pmatrix} \F'\Xv \\ \V'\Xv \\ \end{pmatrix},
\end{equation}
where $\dGV = \GV + \s_0^2\D^{-1} = \V'\V+\s_0^2\D^{-1}$ and $\GV = \V'\V \in \Rm{l}{l}$ is the 
Gram 
matrix (\cite[sec.~10.1]{boyd_introduction_2018}) for columns of design matrix $\V$.

Under the assumption of the full column rank $r= k+l$ of matrix $(\F \;\V)$, when FDSRLM is  
identifiable, and using the so-called \textit{Banachiewicz inversion formula} for the inverse of a  
$2\times 2$  partitioned (block) matrix (see e.g. \cite[sec.~6.0.2]{zhang_schur_2005} or 
\cite[sec.~2.1]{hancova_natural_2008}), we can easily prove the existence and the following form  
of the inverse to the block matrix in \eqref{eq:MME}\\
$$\blockmatrix{\F'\F}{\F'\V}{\V'\F}{\dGV}^{-1} = 
\blockmatrix{(\F'\F)^{-1}}{\Om{k}{l}}{\Om{l}{k}}{\Om{l}{l}}+
\begin{pmatrix} -(\F'\F)^{-1} \F'\V \\ \Il \\ \end{pmatrix}
\dW^{-1}\left( -\V'\F (\F'\F)^{-1}  \; \Il \right),
$$\\
where $\dW = \W + \s_0^2\D^{-1} = \V'\Mf\V+ \s_0^2\D^{-1} \succ 0$ is again the Schur complement of 
$\F'\F$ but now in the block matrix $\smallblockmatrix{\F'\F}{\F'\V}{\V'\F}{\dGV}$. We also call 
matrix $\dW$ more specifically as \textit{the extended Schur complement} determined by variances 
$\nuv$. 

Substituting the last result for the inverse into MME \eqref{eq:MME} and rearranging, we get for  
$\blupYn$ (symbols $\bullet$ denote blocks not needed in deriving)\\
$$\begin{pmatrix} \bluebn \\ \blupYn \\ \end{pmatrix} = 
{\blockmatrix{\F'\F}{\F'\V}{\V'\F}{\dGV}}^{-1} \cdot
\begin{pmatrix} \F'\Xv \\ \V'\Xv  \end{pmatrix} =
\blockmatrix{\bullet}{\bullet}{-\dW^{-1}\V'\F(\F'\F)^{-1}}{\dW^{-1}}\cdot
\begin{pmatrix} \F'\Xv \\ \V'\Xv \\ \end{pmatrix}$$

\begin{equation}\label{eq:blupY}
	\blupYn = 
	\dW^{-1}\V'(\In-\F(\F'\F)^{-1}\F')\Xv=
	\dW^{-1}\V'\Mf \Xv
\end{equation}
Denoting matrix $\dW^{-1}\V'\Mf \in \Rm{l}{n}$  as $\dT$, we just derived the following 
computational form of BLUP-NE estimators.

\begin{proposition}[the computational form of BLUP-NE]
\leavevmode\\
Let us consider the following LMM model for FDSLRM observation $\Xv$
\begin{gather*}
	\Xv=\F \bv +\V \Yv+\wv, \quad \E{\wv}=\Ov{n}, \quad \Cov{\wv}=\s_0^2\In, \\
	\Cov{\Yv}=\diag{\s_j^2},\quad \Cov{\Yv,\wv}=\Om{l}{n}.
\end{gather*}
Then BLUP-NE estimators $\bess{1},\ldots,\bess{l}$ of parameters $\s_1^2,\ldots,\s_l^2$ 
are given by
\begin{equation} \label{eq:compformessj}
	\besX{j}={(\blupY)}^2_j=(\dT \Xv)_j'{(\dT \Xv)}_j
	=\Xv'\dt_j\dt_j'\Xv;\text{ }j=1,\ldots,l,
\end{equation}
where $\dt_j=(\dT_{j1},\dT_{j2},\ldots,\dT_{jn})'$ are rows of matrix $\dT$ and 
$\dT=\dW^{-1}\V'\Mf$ with $\dW= \diag{\s_0^2/\s_j^2}+\V'\Mf\V$. \\
\end{proposition}

\begin{remark} 
\leavevmode\\
\small	
It is straightforward to extend our computational form to the case when some 
variances $\s_j^2$ are zero or in other words when matrix $\D$ is singular. 	
For a singular $\D$, MME \eqref{eq:MME} have the alternative form \cite{witkovsky_estimation_2012}
\begin{equation} \label{eq:MME2}
	\blockmatrix{\F'\F}{\F'\V\D}{\V'\F}{\dHV } \cdot
	\begin{pmatrix} \bluebn \\ \blupZn \\ \end{pmatrix} = 
	\begin{pmatrix} \F'\Xv \\ \V'\Xv \\ \end{pmatrix},
\end{equation}
where $\dHV = \GV\D  + \s_0^2 \Il = \V'\V\D+\s_0^2 \Il $ and $\blupY = \D\blupZ$.

Applying the same argument with the Banachiewicz inversion formula, we get the second version of 
the computational form of $\dT$ determining $\blupYn$
\begin{equation} \label{eq:compformessj2}
 \dT= \D \dU^{-1} \V'\Mf \text{ with } \dU=\W\D+ \s_0^2\Il = \V'\Mf\V \D + \s_0^2\Il.
\end{equation}
As opposed to $\dW, \dGV$, matrices $\dU, \dHV$ always exist under identifiability assumptions of 
our FDSLRM. 
Simultaneously both of $\dU, \dHV$ are also always invertible, since both of $\det \dU, \det \dHV$ 
are nonzero for any $\D$ in our FDSLRM as it is shown in the Appendix. In the case of a nonsingular 
$\D$ the mentioned matrices are connected via following relationships
\begin{equation}\label{eq:HV GV W U}
	\dHV = \dGV \D, \; \dGV^{-1} = \D \dHV^{-1}, \quad \dU = \dW \D, \; \dW^{-1} = \D\dU^{-1}.
\end{equation}
The version \eqref{eq:MME2} of MME is also preferred in numerical calculations 
\cite{witkovsky_estimation_2012}, since it can handle not only a singular $\D$ but also a very 
ill-conditioned $\D$ appearing when  $\nuv$ has very small positive components $\s_j^2$. 

It is worth to mention that originally MME were derived under the normality assumptions (see the 
original work \cite{henderson_estimation_1959} or \cite{witkovsky_estimation_2012}), but from the 
viewpoint of least squares both versions \eqref{eq:MME}, \eqref{eq:MME2} of MME describe the BLUP 
$\blupYn$ and BLUE $\bluebn$ (\cite[sec.~12.3]{christensen_plane_2011}) with no need to restrict 
distributions of $\Yv$ and $\wv$ to be normal.\\[-6pt]
\end{remark} 

\subsection{Statistical properties at known variance parameters}
\label{sec:properties}

The derivation of theoretical properties of BLUP-NE estimates under the assumption of known $\nuv$, 
regarding the first and second order moment characteristics, can be significantly facilitated by 
the following lemma describing the properties of matrix $\dT$ determining the BLUP $\blupYn$ in 
\eqref{eq:blupY}. 

Its proof can be accomplished in a similar way as the proof of Lemma 3.1 in 
\cite{hancova_natural_2008} by a direct routine computation employing formulas 
\eqref{eq:HV GV W U}, $\dT=\dW^{-1}\V'\Mf$, $\dW=\V'\Mf\V+\s_0^2\D^{-1}$, properties of orthogonal 
projectors (orthogonality, idempotency, symmetry) and Schur complements (symmetry, positive 
definiteness).

\begin{lemma}[basic properties of $\dT$]\label{thm:propertiesT}
	\begin{enumeratep}
	\item $\dT\dT'=\dW^{-1}(\Il-\s_0^2\D^{-1}\dW^{-1})=\D\dU^{-1}(\Il-\s_0^2\dU^{-1})$ 
	\item $\dT\F=\Om{l}{n}$ and $\dT\V=\Il-\s_0^2\D^{-1}\dW^{-1} = \Il-\s_0^2\dU^{-1}$
	\item $\dT\Sgn\dT'=\D-\s_0^2 \dW^{-1} = \D (\Il - \s_0^2 \dU^{-1})$
	\end{enumeratep}
\end{lemma}

The computational forms \eqref{eq:compformessj} of BLUP-NE are also quadratic forms of $\Xv$. In 
addition, result (2) $\dT\F = \Om{l}{n} $ implies that $\dt_j\F=\Ov{}$. Such a condition leads to 
the conclusion that BLUP natural estimators $\besX{j}$ are translation invariant or shortly 
invariant 
quadratic estimators \cite[sec.~1.5]{stulajter_predictions_2002} or 
\cite[sec.~3.1]{hancova_natural_2008}. The following theorem summarizes theoretical properties of 
BLUP-NE.

\begin{theorem}[statistical properties of $\besX{j}$]\label{thm:propertiesBLUPNE}
\leavevmode \\
Natural estimators $\besX{j};j=1,2, \ldots, l$ of $\nuv$ based on the BLUP $\blupYn$ of $\Yv$ are 
invariant quadratic estimators having the following properties
\begin{enumeratep}
\item $\En{\besX{j}}=\s_j^2-\s_0^2(\dW^{-1})_{jj};$  $j=1,2,\ldots,l,$ 

\vertspace	
\noindent If $\Xv \sim  \NnX$, then
\vertspace	
\item $\Din{\besX{j}}=2(\s_j^2-\s_0^2(\dW^{-1})_{jj})^2;$  $j=1,2,\ldots,l,$
\item $\Covn{\besX{i},\besX{j}}=2(\s_0^2 (\dW^{-1})_{ij})^2;$  $i,j=1,2,\ldots,l,$ $i\neq j$,
\item $\MSEn{\besX{j}}=2{(\s_j^2-\s_0^2(\dW^{-1})_{jj})}^2+\s_0^4(\dW^{-1})_{jj}^2;$ 
$j=1,2,\ldots,l.$
\end{enumeratep}
\end{theorem}

\begin{proof}
	See the Appendix.
\end{proof}

\begin{remark}
\leavevmode\\
\small
If some components in $\nuv$ are zero, there is a need to use expression 
$\D\dU^{-1}$ instead of $\dW^{-1}$. The first property, biasedness of 
$\besX{j}, j = 1,2,\ldots,l$, is in full accordance with the Ghosh theorem  
(\cite{ghosh_nonexistence_1996}, 
\cite{gajdos_kriging_2017}) about the incompatibility between simultaneous non-negativity and 
unbiasedness of estimators for variances of random components $\Yv$ in LMM. As for bias, defined as 
$\biasn{\besX{j}} = \En{\besX{j}}-\s_j^2$, we have the following expression given by the extended 
Schur complement $\dW$

\begin{equation}
\biasn{\besX{j}} = -\s_0^2(\dW^{-1})_{jj}
\end{equation}
\end{remark}

Now we can theoretically compare quality of the BLUP natural estimators with original natural 
estimators (NE) proposed in our previous paper \cite{hancova_natural_2008}. As we can see from 
summarizing tab. \ref{tab:comparison}, their quality is determined by Schur complements $\W$ and 
$\dW$ where both of them are symmetric and positive definite matrices.

\begin{table}[H]
\renewcommand{\arraystretch}{1.25}
\begin{center}
\caption{First and second order moment characteristics for BLUP-NE and NE.}
\label{tab:comparison}
\begin{tabular}{l r r}
\toprule[1.5pt]
\small \textbf{characteristic }&\textbf{NE estimators $\nesX{j}$} &\textbf{BLUP-NE estimators 
$\besX{j}$}\\
\midrule
\normalsize
$\Updelta_{\nuv}$ & $\s_0^2 (\W^{-1})_{jj}$              
			      & $-\s_0^2 (\dW^{-1})_{jj}$ \\
$D_{\nuv}$        & $2(\s_j^2+\s_0^2(\W^{-1})_{jj})^2\phantom{+\s_0^4(\W^{-1})_{jj}^2x}$ 
                  & $2(\s_j^2-\s_0^2(\dW^{-1})_{jj})^2\phantom{+\s_0^4(\dW^{-1})_{jj}^2x}$\\ 
$MSE_{\nuv}$      & $2(\s_j^2+\s_0^2(\W^{-1})_{jj})^2+\s_0^4(\W^{-1})_{jj}^2$
                  & $2{(\s_j^2-\s_0^2(\dW^{-1})_{jj})}^2+\s_0^4(\dW^{-1})_{jj}^2$\\
\bottomrule[1.5pt]
    \end{tabular}
\end{center}
\end{table}

Using elementary properties of the L\"owner partial ordering 
\cite[chap.~24]{puntanen_formulas_2013}, we have immediately relations $\dW > \W,  
\W^{-1}>\dW^{-1}, (\W)_{jj} > (\dW^{-1})_{jj}$ which imply a smaller absolute value of bias, 
dispersion and MSE for BLUP-NE $\besX{j}$ in comparison with original NE $\nesX{j}$.  This 
conclusion directly 
supports our heuristic idea leading to the BLUP-NE definition \eqref{def:BLUPNE}. 

\subsection*{\textit{Orthogonal FDSLRM}} \label{sec:orthogonal}

In real time series analysis using FDSLRMs, the fundamental modeling 
procedure is based on spectral analysis of time series (\cite{priestley_spectral_2004}, 
\cite{brockwell_time_2009}). We identify the significant Fourier frequencies by periodogram 
(\cite{stulajter_estimation_2004}, \cite{gajdos_kriging_2017}) which restrict the form of design 
matrices $\F$ an $\V$ leading to the orthogonality condition for FDSLRM 
\cite{stulajter_estimation_2004} 
\begin{equation}\label{eq:orthogonality}
\F'\V = 0  \text{ and } \GV = \V'\V = \diag{\norm{\vv_j}^2}.
\end{equation}
where $\vv_j, j = 1,2\ldots,l$ is \textit{j}-th column of $\V$. Such a FDSLRM, satisfying condition 
\eqref{eq:orthogonality}, is called \textit{orthogonal} \cite{stulajter_estimation_2004}. So in 
practice, we mainly work with the orthogonal version of FDSLRM. 

Under the condition of orthogonality 
\eqref{eq:orthogonality}, we can write for matrices $\Mf$, $\dU$, 
$\dW$, $\dT$ and BLUP-NE $\besX{j}$
\begin{gather*}
\Mf \V=\V, \; \dU = \GV \D +\s_0^2 \Il = \diag{\s_0^2+\s_j^2\norm{\vv_j}^2}, \\
\dW^{-1} =\D \dU^{-1} = \diag{\dk_j/\norm{\vv_j}^2}, \; \dt_j = \dk_j\vv_j'/\norm{\vv_j}^2, \\
\besX{j} = \dk_j^2 \, \Xv'\vv_j\vv_j'\Xv/\norm{\vv_j}^4, \; j = 1,2, \ldots, l,
\end{gather*}
where we introduced $ \dk_j 
\equiv\frac{\s_j^2\norm{\vv_j}^2}{\s_0^2+\s_j^2\norm{\vv_j}^2} \in \Rv{}, \; 0\leqq \dk_j <1, j = 
1, 2, \ldots, l $. \\

Applying these results, we get the following direct corollary of Theorem 
\ref{thm:propertiesBLUPNE} for any orthogonal FDSLRM.

\begin{corollary}[properties of $\besX{j}$ in orthogonal FDSLRM]\label{cor:essj}
\leavevmode\\
In an orthogonal FDSLRM natural estimators $\besnX{j}$ based on the BLUP $\blupYn$ 
have the following properties
\begin{enumeratep}
		\item $\En{\besX{j}}= \dk_j \s_j^2 $; $j=1,2,\ldots,l,$ where 
		$\dk_j=\dfrac{\s_j^2\norm{\vv_j}^2}{\s_0^2+\s_j^2\norm{\vv_j}^2} \in \Rv{}_+$. 
		
		\noindent If $\Xv \sim  \NnX$, then 
		\vertspace		
		\item $\Din{\besX{j}}= 2 \dk_j^2 \s^4_j$; $j=1,2,\ldots,l,$
		\item $\Covn{\besX{i},\besX{j}}=0$;  $i,j=1,2,\ldots,l,$ $i \neq j$,
		\item $\MSEn{\besX{j}}=[2\dk_j^2+(1-\dk_j)^2]\s^4_j$;	$j=1,2,\ldots,l.$
\end{enumeratep}
\end{corollary}
  
Finally, if we look at original natural estimators $\nesX{j}$ from our paper 
\cite{hancova_natural_2008}, then in orthogonal FDSLRM it can be easily shown that $\besX{j} = 
\dk_j^2 \nesX{j}, j = 1, 2 , \ldots, l$. If we introduce $\rho_0 = 1 $, then we obtain the complete 
orthogonal version of \eqref{eq:compformessj} for computing BLUP-NE $\besX{j}$:
\begin{equation} \label{eq:compformessjort}
\besX{j} = \dk_j^2 \nesX{j}, j = 0, 1 , \ldots, l.
\end{equation}

\Section{Mathematical and computational tools of convex optimization}

\subsection{Empirical BLUPs in the BLUP-NE estimation method}

The computational form of our BLUP natural estimators, their first and second order moment 
characteristics were derived at known variances $\nuv$. But if we really knew variances $\nuv$, 
we would not need to estimate them. The practical estimation procedure for $\nuv$ cannot 
depend on the parameters $\nuv$ which are to be estimated. 

However, such situation is not rare at all in LMMs, also in the FDSLRM theory, when we consider 
e.g. double weighted least squares estimators (DOWELSE) or maximum likelihood estimators (MLE) for 
$\nuv$. In FDSRLM fitting these estimators are computed by iterative numerical procedures using
projections in Hilbert spaces (\cite[sec. 3.4]{stulajter_predictions_2002} or 
\cite{stulajter_estimation_2004}). Generally, in the LMM framework, similar iterative numerical 
estimation procedures are described in \cite[chap.~8]{searle_variance_2009}, 
\cite{witkovsky_estimation_2012}.

In our case, the simplest reasonable solution of the mentioned situation is to use an empirical 
version of BLUP (EBLUP) \label{abv-EBLUP} based on the computational form of BLUP-NE 
\eqref{eq:compformessj}, or 
\eqref{eq:compformessjort} in an orthogonal case, as it is usual in the general theory of empirical 
BLUPs 
in LMMs (\cite[chap.~5]{stulajter_predictions_2002}, 
\cite{witkovsky_estimation_2012}, 
\cite[chap.~5]{rao_small_2015}). In particular, this step means 
replacing unknown true parameters $\nuv$ with other "initial" values or estimates of 
$\nuv$ in FDSLRM. In the light of iterative approaches, our EBLUP-NE \label{abv-EBLUPNE} can be 
viewed as a 
two-stage iterative method with one step in the iteration. 

If we think again heuristically as during the formulation of the BLUP-NE definition 
\eqref{def:BLUPNE}, it seems reasonable to expect as good starting values of $\nuv$  any, 
in some sense optimal estimates of $\nuv$, e.g. estimates obtained by maximum likelihood or least 
squares in FDSLRM. 

On the top of that, from the theoretical perspective we should prefer estimates with the simplest 
explicit form and under 
less restrictive assumptions on the structure or distribution of FDSLRM. On the other hand, from 
the practical point of 
view, it will be sufficient to have at least initial estimates which can be obtained by reliable 
and time efficient 
computational methods.

Therefore, in the following sections we investigate, theoretically and practically with real data 
sets, five estimation methods providing initial estimates for $\nuv$ under different assumptions on 
the structure and distribution of FDSLRM observation $\Xv$. Our candidates for the second stage of 
EBLUP-NE are summarized in tab. \ref{tab:methods}.

\vspace{-15pt}

\begin{table}[H]
	\renewcommand{\arraystretch}{1.25}
	\begin{center}
	\caption{\small Estimation methods chosen for the second stage in EBLUP-NE.}
	\label{tab:methods}
	\begin{tabular}{c c}
		\toprule[1.5pt]
		\textbf{least squares method in FDSLRM} & \textbf{maximum likelihood 
		method in FDSLRM}\\
		\midrule
		\small
		double ordinary least squares estimators & maximum likelihood estimators  \\
		(DOOLSE) ref. \cite{stulajter_estimation_2004}, \cite{stulajter_predictions_2002} 
		& (MLE) ref. \cite{stulajter_estimation_2004}, \cite{stulajter_predictions_2002}  \\ 
		\midrule
		modified (unbiased) DOOLSE	& restricted (residual) MLE \\
		(MDOOLSE) ref. \cite{stulajter_estimation_2004}
		& (REMLE) ref. \cite{stulajter_estimation_2004} \\
		\midrule
		natural estimators (NE) ref. \cite{hancova_natural_2008} & \\
		\bottomrule[1.5pt]
	\end{tabular}
	\end{center}
\end{table}

\vspace{-15pt}

Since theoretical and computational aspects of these 
estimation procedures in FDSLRM fitting based on the projection theory in Hilbert spaces are more 
than 10 years old (\cite{stulajter_predictions_2002}, \cite{stulajter_estimation_2004}, 
\cite{hancova_natural_2008}), we revisit and update them in the light of the current theoretical 
and computational tools of convex optimization. 

According to \cite{agrawal_rewriting_2018}, \cite{boyd_convex_2009}, to formulate any mathematical 
problem as an optimization problem, we need to identify three attributes of the problem: 
optimization variable, constraints that the variable must satisfy and the objective function 
depending on the variable whose optimal value we want to achieve. 

In all our estimation methods, optimization variable is $\nuv \in \Rv{l+1}$ satisfying  
constraints $\Ys = \Rv{}_{++} \times \Rv{l}_{+} $. In order to avoid any incompatibility problems 
in the convex optimization theoretical or computational framework, we extend constraints  $\Ys$ 
into the form of standard nonnegativity constraints $\Ys^{*} = [0, \infty]^{l+1}$ or  $\nuv\succeq 
0$ using generalized inequality.

Finally, we also consider using present theoretical and computational tools used in linear mixed 
modeling, since in the LMM framework, the problem of estimating variances has a long and rich 
history with many essential reference works (e.g. \cite{rao_estimation_1988}, 
\cite{searle_variance_2009}, \cite{christensen_plane_2011}, \cite{demidenko_mixed_2013}, 
\cite{rao_small_2015}, see also a review paper \cite{witkovsky_estimation_2012}). 

\subsection{General case of the orthogonal FDSLRM}
First of all, we formulate the basic assumptions required in our investigations on the structure 
and 
distribution of FDSLRM. As we mentioned  in subsection \ref{sec:orthogonal}, in real (practical) 
FDSLRM analysis we mainly work with orthogonal FDSLRMs \eqref{eq:orthogonality}. Therefore, in the 
rest of the paper we will focus primarily on this type of FDSLRMs. 

The second assumption deals with a distribution of the FDSLRM observation $\Xv = 
(X(1),\ldots,X(n))',$ $n \in \N$. For now, we assume $\Xv$  satisfying \eqref{eq:FDSLRMobservation} 
and having any absolutely continuous probability distribution with respect to some 
$\sigma$-finite measure.

Under these two assumptions, we can apply in the second stage of EBLUP-NE three estimation methods: 
\textit{natural estimators} (NE), \textit{double ordinary least squares estimators} (DOOLSE) and 
\textit{modified double ordinary least squares estimators} (MDOOLSE).

\vspace{-6pt}

\paragraph{Original natural estimators -- NE} 
\leavevmode\\[6pt]
In the case of NE, we know the analytic solution of the estimation problem 
\cite{hancova_natural_2008}, which gives us always required non-negative estimates of FDSLRM 
variances. 
Employing results of 
\cite{hancova_natural_2008}, orthogonality condition \eqref{eq:orthogonality} leads to the 
following form of NE
\begin{equation}\label{eq:orthoNE}
\renewcommand{\arraystretch}{1.4}
\breve\nuv(\epsv) =
\begin{pmatrix}
\tfrac{1}{n-k-l}\,\epsv'\,\Mv\,\epsv \\
(\epsv'\vv_1)^2/\norm{\vv_1}^4 \\
(\epsv'\vv_2)^2/\norm{\vv_2}^4 \\
\vdots \\
(\epsv'\vv_l)^2/\norm{\vv_l}^4
\end{pmatrix} = 
\begin{pmatrix}
\tfrac{1}{n-k-l}\left(\epsv'\epsv-\Sum{j}{l}(\epsv'\vv_j)^2/\norm{\vv_j}^2\right) \\
(\epsv'\vv_1)^2/\norm{\vv_1}^4 \\
(\epsv'\vv_2)^2/\norm{\vv_2}^4 \\
\vdots \\
(\epsv'\vv_l)^2/\norm{\vv_l}^4
\end{pmatrix},
\end{equation}
where $\epsv=\xv-\F\estbv =\Mf\xv$ is nothing else than the vector of ordinary least squares (OLS) 
\label{abv-OLS} residuals in FDSLRM, $\Mv = \In - \V(\V'\V)^{-1}\V'$ is the orthogonal projector 
onto the orthogonal element of the column space of $\V$ and $\xv$ is an arbitrary realization of 
the FDSLRM observation $\Xv$. Vectors $\vv_j, j = 1, 2, \ldots, l$ are columns of design matrix 
$\V$. Moreover, in orthogonal FDSLRM BLUE  $\bluebn= 
(\F'\InvSgn\F)^{-1}\F'\InvSgn\xv$ from MME \eqref{eq:MME} is identical with ordinary least squares 
estimate $\estbv =(\F'\F)^{-1}\F'\xv $ not depending on $\nuv$ \cite{stulajter_estimation_2004}.

As for computational complexity of NE, using elementary theory of complexity 
\cite{boyd_introduction_2018}, we get the complexity $2kn$ operations for the residual 
vector $\epsv$ and $4ln$ for NE from \eqref{eq:orthoNE}.  Therefore, computing NE 
represents an algorithm with the complexity having order $\bigO(n)$ 
with respect to the realization length $n$. 

\begin{remark}
\small
\leavevmode\\
For the purposes of software cross-checking of results and evaluating numerical precision of convex 
optimization algorithms, it is worth to formulate the calculation of NE as a convex optimization 
problem. Since NE employ the least-squares method, it is not complicated to show that in 
orthogonal FDSLRMs,  NE can be obtained as a unique, always existing, 
non-negative solution of the following convex optimization problem

\begin{flalign}
\qquad &
\begin{array}{ll} 
\textbf{NE} & \\
\emph{minimize}    & \quad 
f_0(\nuv)=||\epsv\epsv' - \Sum{j}{l}\nus_j\V_j||^2+\norm{\Mv\epsv\epsv'\Mv-\nus_0\V_0}^2 \\[6pt]
\emph{subject to}  & \quad \nuv = \left(\nus_0, \ldots, \nus_l\right)'\in [0, \infty)^{l+1} 
\end{array}
&& \label{def:cvxNE} 
\end{flalign}
where matrices in the objective function $f_0(\nuv)$ are defined by expressions
\begin{equation} \label{eq:MGVVj}
\V_0=\Mf\Mv \text{ and } \V_j = \vv_j\vv_j', j=1,\ldots,l.
\end{equation}
In convex optimization, this kind of optimization problems belongs to the  convex quadratic 
optimization problems or the norm approximation problems \cite[sec.~4.4, 
sec.~6.1]{boyd_convex_2009}.
\end{remark}
\paragraph{Double least squares estimators -- DOOLSE, MDOOLSE} 
\leavevmode\\[6pt]
From the optimization viewpoint, assuming general (not necessarily orthogonal) FDSLRM observation 
$\Xv$ \eqref{eq:FDSLRMobservation}, DOOLSE and MDOOLSE (\cite{stulajter_predictions_2002}, 
\cite{stulajter_estimation_2004}) can be 
viewed as the following optimization problems for $\nuv$ at given OLS residuals $\epsv$ 
(\cite{agrawal_rewriting_2018}, \cite{boyd_convex_2009})
\begin{flalign}
\qquad &
\begin{array}{ll} 
\textbf{DOOLSE\phantom{M}} & \\[6pt]
\emph{minimize}    & f_0(\nuv)=\norm{\epsv\epsv'-\Sgn}^2 \\[6pt]
\emph{subject to}  & \nuv = \left(\nus_0, \ldots, \nus_l\right)'\in [0, \infty)^{l+1} 
\end{array}
&& \label{def:DOOLSE} \\[12pt]
\qquad &
\begin{array}{ll} 
\textbf{MDOOLSE} & \\[6pt]
\emph{minimize} & f_0(\nuv)=\norm{\epsv\epsv'-\Mf\Sgn\Mf}^2 \\[6pt]
\emph{subject to} & \nuv = \left(\nus_0, \ldots, \nus_l\right)'\in  [0, \infty)^{l+1}
\end{array}
&& \label{def:MDOOLSE}
\end{flalign}
where $\epsv$ are again OLS residuals as in the case of NE and $\Sgn = \s_0^2\In + \V\D\V'$ is the 
covariance matrix of $\Xv$. We call matrix $\Seps \equiv \epsv\epsv'$  \textit{matrix of residual 
products} or more compactly \textit{the residual products matrix}. 

\begin{remark} \label{rem:(M)DOOLSE}
	\small
	\leavevmode\\
	In the next text, some theoretical results can be written in one form for both of DOOLSE and 
	MDOOLSE problems. To emphasize this fact, we will use one abbreviation with parenthesis $(~)$, 
	e.g (M)DOOLSE, assuming  that given results or considerations hold for both problems.
\end{remark}

Using the geometrical language of 
Hilbert spaces (\cite{brockwell_time_2009}, \cite{stulajter_predictions_2002}), (M)DOOLSE 
determine a covariance matrix with the parametric structure as in $\Sgn$ but having the smallest 
Euclidean distance to the matrix $\Seps$. In such case,  (M)DOOLSE for orthogonal FDSLRMs could 
be computed geometrically as the orthogonal projection of $\Seps$ onto the linear span $\LS{\V_0, 
\V_1,\ldots, \V_l}$ using the Gram matrix $\G$ of $\{\V_0, \V_1, \ldots, 
\V_l\}$ generated by the inner 
product $(\bullet,\bullet) = \tr(\bullet \cdot \bullet)$. Matrices $\V_j, j=1,\ldots, l$ are 
given by \eqref{eq:MGVVj}, for DOOLSE:  $\V_0 = \In$ and for MDOOLSE: $\V_0 = \Mf$.

According to \cite{stulajter_estimation_2004}, \cite{gajdos_kriging_2017}, the mentioned projection 
is given by 
\begin{equation} \label{eq:projectionMDOOLSE}
\est\nuv(\epsv) = \G^{-1}\qv,
\end{equation}
where 
$ \quad \G = \begin{pmatrix}
n^* & \norm{\vv_{1}}^2 & \norm{\vv_{2}}^2 & \ldots & 
\norm{\vv_{l}}^2 \\
\norm{\vv_{1}}^2 & \norm{\vv_{1}}^4 & 0 & \ldots & 0 \\
\norm{\vv_{2}}^2 & 0 & \norm{\vv_{2}}^4 & \ldots & 0 \\
\vdots & \vdots & \vdots & \ldots & \vdots \\
\norm{\vv_{l}}^2 & 0 & 0 & \ldots & \norm{\vv_{l}}^4
\end{pmatrix}\succ0 
$,$\qquad \qv = 
\begin{pmatrix}
\epsv'  \epsv\\
(\epsv'  \vv_{1})^2 \\
(\epsv'  \vv_{2})^2 \\
\vdots \\
(\epsv'  \vv_{l})^2
\end{pmatrix},\\[12pt]$
where for DOOLSE:  $n^* \equiv n$, for MDOOLSE: $n^* \equiv n-k$, $\epsv$ 
are again OLS residuals, $\vv_j$ are columns of $\V$.

However, the projection method can produce estimates out of the parametric space $\Ys$ at a 
relatively high probability \cite{gajdos_kriging_2017}, i.e. in many cases it produces negative 
estimates. The proposed method \eqref{eq:projectionMDOOLSE} is not able to handle constraints given 
by $\Ys$ reliably and in 
some simple way. 

To distinguish the projection estimates from real DOOLSE \eqref{def:DOOLSE}, MDOOLSE 
\eqref{def:MDOOLSE} which are always non-negative, 
we suggest to be more accurate and use full names for real DOOLSE and MDOOLSE: 
\textit{non-negative} 
DOOLSE (NN-DOOLSE) and \textit{non-negative} MDOOLSE (NN-MDOOLSE) 
This specific way is similarly used e.g. in the case of Rao's MINQUE \cite[chap. 
5]{rao_estimation_1988}. For the projection-based DOOLSE and MDOOLSE, without considering
nonnegativity constraints, we leave the original acronyms DOOLSE, MDOOLSE.

Applying the basic convex quadratic optimization theory (\cite{cornuejols_optimization_2018}, 
\cite{bertsekas_convex_2009}) in orthogonal FDSLRM, we can rewrite NN-(M)DOOLSE as strictly convex 
quadratic problems, whose solutions $\est\nuv$ always exist and are unique global minimizers on 
$\Ys^{*}$ (see Propositions \ref{pro:cvxMDOOLSE} and \ref{pro:cvxMDOOLSEexist} in the Appendix)
\begin{flalign}
\qquad&
\begin{array}{ll} 
\textbf{NN-(M)DOOLSE} & \\[6pt]
\emph{minimize}  & \quad f_0(\nuv)= \nuv'\G\nuv - 2\qv'\nuv\\[6pt]
\emph{subject to} & \quad -\I{l+1}\nuv \preceq \Ov{l+1}
\end{array} 
&&
\end{flalign}

The most important result of convex optimization for NN-(M)DOOLSE problems are 
fundamental \textit{Karush-Kuhn-Tucker (KKT) optimality conditions}\label{abv-KKT}, which in 
orthogonal FDSLRM 
provide a \textit{necessary and sufficient} condition for optimality of solutions (see Proposition 
\ref{pro:KKTcvxMDOOLSE} in the Appendix). In practice, finding solutions of convex optimization 
problems is nothing else than solving the corresponding KKT conditions analytically or in general 
numerically. 

In our case of NN-(M)DOOLSE, the KKT conditions can be represented by a set of $2^{l}$ 
nonsingular linear systems of equations, where each of them is given by matrix $\K$
derived from  $\G$ and by vector $\qv$ from \eqref{eq:projectionMDOOLSE}. Our theoretical 
results in the Appendix show that only one of the linear systems always gives us all required 
non-negative elements of $\est\nuv$. 

On this basis, we can establish a simple KKT optimization 
algorithm which run through given linear systems, compute their analytical solution and stops when 
the solution appears non-negative. In other words, the proposed KKT algorithm always founds the 
required NN-(M)DOOLSE  $\nndest{\nuv}$ in at most $2^{l}$ steps. The algorithm is summarized in 
tab. \ref{tab:KKTalgorithm}, whereas the proof and details are in the Appendix. 

As for computational complexity of the KKT optimization algorithm, all 
calculations of input have complexity $\bigO(n)$, whereas the body of the algorithm (steps 1-3) is 
$\bigO(l^2\cdot2^{l})$. Since fixed number $l$ of variance parameters in matrix $\D_{\nuv} = 
\Cov{\Yv}$ is usually much smaller than $n$ (typically $l^2\cdot2^l$ is of $\bigO(n)$), then the 
complete algorithm has the leading order $\bigO(n)$ with respect to $n$.

\begin{table}[H]
\begin{center}
\caption{\small Scheme of the KKT algorithm for NN-(M)DOOLSE $\nndest{\nuv}$ in orthogonal FDSLRM.}
\label{tab:KKTalgorithm}
		\begin{tabular}{p{12cm}}
			\toprule[1.5pt]  \\
			\textbf{Input:} Form the matrix $\G$, the  vector $\qv$ from vectors $\epsv$ and 
			$\vv_j; j=1,\ldots,l$. \\[6pt]
			\noindent \textbf{For} each auxiliary vector $\bvv =(b_1, b_2, \ldots, b_l)' \in 
			\{0,1\}^{l}$ \textbf{do}: 
			\begin{enumerate}
				\item\textit{Set the KKT-conditions matrix} $\K$: 
				\item[]	$\K \leftarrow \G$; 
				\item[]\textbf{For} $j=1,2,\ldots,l$ \textbf{do}: \textbf{If} 
				$b_j=0$ \textbf{then} 
				$K_{0j}\leftarrow 
				0, K_{jj} \leftarrow -1$.
				\item\textit{Calculate the auxiliary vector} $\gv$:
				\item[] $\gv \leftarrow \K^{-1}\qv$. 
				\item \textit{Test non-negativity of}  $\gv$: 
				\item[] \textbf{If} $\gv \geq 0$ \textbf{then} \textbf{quit}.
			\end{enumerate} \\
			\textbf{Output:} Use the last $\bvv, \gv$ to form NN-(M)DOOLSE $\nndest{\nuv}$ of 
			$\nuv$: \\
			\hspace{4ex} $\nndest{\nuv}\leftarrow \gv$; \\
			\hspace{4ex} \textbf{For} $j=1,2,\ldots,l$ \textbf{do}: 
			\textbf{If} $b_j=0$ \textbf{then} 
			$\est\nus_{+j}\leftarrow 0$. \\[6pt]
			\bottomrule[1.5pt]
		\end{tabular}
\end{center}
\end{table}

\noindent It is the same as in the case of NE and generally typical only for computationally 
fastest analytical solutions of the KKT conditions. 

\paragraph{Computational tools for orthogonal FDSLRM} 
\leavevmode\\[6pt]
In nonlinear optimization, there is a variety of highly efficient, fast and reliable 
open-source and commercial software packages \cite[chap.~20]{cornuejols_optimization_2018}. We have 
chosen one of the most well-known open source libraries for solving convex optimization tasks -- 
CVXPY \cite{diamond_cvxpy:_2016} and its R version CVXR \cite{fu_cvxr:_2019}. 

CVXPY \label{abv-CVXPY} is a scientific Python library but also a language with very simple, 
readable syntax not 
requiring any expertise in convex optimization and its PC implementation. CVXPY allows the user to 
specify the mathematical optimization problem naturally following normal mathematical notation as 
we can see in computing NN-DOOLSE (fig.~\ref{fig:doolsecvxpy}) where a code easily mimics the 
non-negative DOOLSE mathematical formulation \eqref{def:DOOLSE} with $\Sgn = 
\s_0^2\In+\V\D\V'=\nus_0\In+\V\diag{\nus_j} \V'$.

CVXPY, implements not only convex optimization  solvers using interior-point numerical methods 
which are extremely reliable and fast, but also verifies convexity of the given problem 
using rules of disciplined convex programming \cite{grant_disciplined_2006}. In FDSLRM, 
interior-points numerical methods have complexity $n^3$ for one iteration or $\log(1/\eps)$ times 
bigger for the complete computation with a  precision $\eps$ of the required optimal solution 
\cite{boyd_convex_2009}.

\begin{remark}
	\small
	\leavevmode\\
	It is worth to mention that both CVXPY and CVXR are able to solve non-negative DOOLSE and 
	MDOOLSE 
	in 
	a general FDSLRM without the orthogonality condition. CVXPY was inspired by MATLAB optimization 
	package CVX \cite{grant_cvx:_2018} still used in many references, e.g. 
	\cite{cornuejols_optimization_2018}. However, in CVXPY (CVXR) the user can easily combine 
	convex 
	optimization and  simplicity of Python (R) language together with its high-level features such 
	as 
	object-oriented design or parallelism. 
\end{remark}

\begin{figure}[H]
	\centering
	\includegraphics[width=0.9\linewidth]{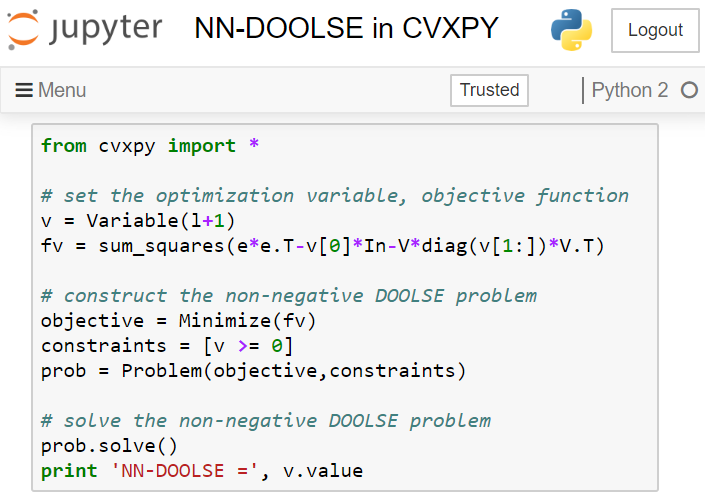}
	\caption{\small CVXPY code for computing DOOLSE in Jupyter environment.}
	\label{fig:doolsecvxpy}
\end{figure}

\vspace{-15pt}

\subsection{Gaussian orthogonal FDSLRM}
In the LMM framework, the most usual $\Xv$ distribution used in practice is represented by the 
multivariate normal (Gaussian) distribution. Therefore, in addition to the orthogonality of FDSLRM, 
we will require in this section the distributional assumption $\Xv \sim \NnX$. We will refer to 
FDSLRM under these assumptions as  to \textit{Gaussian orthogonal FDSLRM}.

In the case of Gaussian orthogonal FDSLRM, we can add to our investigation last two estimation 
methods from tab.~\ref{tab:methods}: \textit{maximum likelihood estimators} (MLE) and 
\textit{residual maximum likelihood estimators} (REMLE). 

\paragraph{Maximum likelihood estimators -- MLE, REMLE} 
\leavevmode\\[6pt]
Both MLE and REMLE of variances $\nuv$ provide estimates maximizing ML and REML 
loglikelihood functions (logarithms of likelihoods).  Using simple and clear arguments from 
Lamotte's paper \cite{lamotte_direct_2007}, we can easily form the ML loglikelihood function 
$\lm(\nuv, \xv)$ and REML loglikelihood function $\lr(\nuv,\xv)$, assuming 
general (not necessarily orthogonal) FDSLRM observation $\Xv$ \eqref{eq:FDSLRMobservation}, as
\begin{gather} \label{eq:loglikelihoods}
\lm(\nuv, \xv) = \tfrac{1}{2} \left(\ln \det (\InvSgn) -\gnorm{\xv-\F\bluebn}{\InvSgn}^2\right) \\
\lr(\nuv, \xv) = \tfrac{1}{2} \left(\ln \det (\InvSgn) - \ln \det 
(\F'\InvSgn\F)-\gnorm{\xv-\F\bluebn}{\InvSgn}^2\right)
\end{gather}
where $\epsv_{\nuv}=\xv-\F\bluebn$ are FDSLRM residuals, known as \textit{marginal 
residuals} in LMM \cite{singer_graphical_2017}, $\xv$ is an arbitrary realization of the FDSLRM 
observation $\Xv$ with covariance matrix $\Sgn$,  $\bluebn$ is BLUE of $\bv$ in MME \eqref{eq:MME} 
and $\gnorm{\bullet}{\InvSgn}^2 = (\bullet)'\InvSgn(\bullet)$ is a generalized 
vector norm given by $\InvSgn$.

\begin{remark}
	\leavevmode\\
	\small
	In Lammote's paper \cite{lamotte_direct_2007}, which presents a direct derivation of REML 
	likelihood function in LMM using familiar linear algebra operations, we can find that the 
	formulation of ML and REML likelihood in LMM under the assumption of normality was originally 
	done 
	in Harville's works	\cite{harville_bayesian_1974}, \cite{harville_maximum_1977}.  However, 
	LMM references up to 2007 show the explicit form of REML likelihood for LMM rarely. In 
	\v{S}tulajter's FDSLRM reference works dealing with MLE \cite{stulajter_predictions_2002}, 		
	\cite{stulajter_estimation_2004} REML likelihood for FDSLRM is also absent. According to 
	Lammote \cite{lamotte_direct_2007}, the reason why seems to be Harville's too sophisticated, 
	difficult and indirect derivation. 
\end{remark}

Maximum likelihood and residual maximum likelihood estimation of variance components in LMMs 
produce, in general, no analytical expressions for the estimators 
\cite[chap.~6]{searle_variance_2009}. They exist only for $\nus_0 >0$. According to \cite[sec. 
2.5, thm 4]{demidenko_mixed_2013} MLE and REMLE estimates do not exist in general FDSLRM with LMM 
observation $\Xv$ \eqref{eq:FDSLRMobservation} if and only if a realization $\xv$ of $\Xv$ belongs 
to $\LS{\F,\V}$, linear span of columns $\F$ and $\V$. This condition is also equivalent with 
$\epsv \in \LS{\V}$.  

In orthogonal FDSLRM, the MLE and REMLE can be rewritten as the optimization problems for $\nuv$ 
(\cite{agrawal_rewriting_2018}, 
\cite{boyd_convex_2009}) at 
given OLS residuals $\epsv$ (for orthogonal FDSLRM $\epsv = \epsv_{\nuv}$, 
\cite{stulajter_estimation_2004}) 
\begin{flalign}
\qquad &
\begin{array}{ll} 
\textbf{MLE}    &\\[6pt]
\emph{maximize}   & \quad f_0(\nuv)=\ln \det(\InvSgn) - \gnorm{\epsv}{\InvSgn}^2\\[6pt]
\emph{subject to} & \quad \nuv = \left(\nus_0, \ldots, \nus_l\right)'\in [0, \infty)^{l+1} 
\end{array} 
&& \label{def:MLE}
\end{flalign}

\vspace{-6pt}

\begin{flalign}
\qquad &
\begin{array}{ll} 
\textbf{REMLE}    & \\[6pt]
\emph{maximize}   & \quad f_0(\nuv)=\ln  \det(\InvSgn) - \ln\det(\F'\InvSgn\F) -  
                    \gnorm{\epsv}{\InvSgn}^2\\[6pt]	
\emph{subject to} & \quad \nuv =  \left(\nus_0, \ldots, \nus_l\right)' \in [0, \infty)^{l+1}
\end{array}
&& \label{def:REMLE}
\end{flalign}
where $\gnorm{\epsv}{\InvSgn}^2 = \epsv'\InvSgn\epsv =  \tr\left(\Seps\InvSgn\right)$ and $\Seps$ 
is the residual products matrix used in  definitions \eqref{def:DOOLSE}, \eqref{def:MDOOLSE} of 
(M)DOOLSE.  

Concerning convex optimization, (RE)MLE problems \eqref{def:REMLE}, \eqref{def:MLE} are generally 
not convex problems. However in orthogonal FDSLRM, according to 
\cite[sec.~4.1.3]{boyd_convex_2009}, using an 
appropriate bijective transformation, we can reformulate \linebreak (RE)MLE as equivalent 
convex optimization problems (see Proposition \ref{pro:cvxREMLE}). 

Employing the theory of convex optimization, we can prove that equivalent (RE)MLE convex problems 
are strictly convex, for $\epsv \not\in \LS{\V}$ have always a unique solution, unique global 
minimizer satisfying corresponding KKT conditions (see Propositions \ref{pro:cvxREMLEexist}, 
\ref{pro:KKTcvxREMLE} and their proofs in the Appendix). These 
conditions are again necessary and sufficient for optimal solutions. 

Combining all our theoretical results, we arrive to the main theoretical result of the section.
It can be shown that for $\nus_0 \neq 0$ (or $\epsv \not\in \LS{\V} $) the KKT conditions for 
NN-(M)DOOLSE are equivalent with the KKT conditions for (RE)MLE and the bijective transformation. 
Or in other words, we can formulate the following theorem.
\begin{theorem}[equivalence between NN-(M)DOOLSE and (RE)MLE] \label{thm:eqDOOLSEMLE}
	\leavevmode\\
	In a Gaussian orthogonal FDSLRM non-negative (M)DOOLSE are almost sure equal to (RE)MLE, i.e.
	$$P\left( \begin{matrix}
	\emph{non-negative} \\
	\emph{DOOLSE} \\
	\end{matrix} \quad = \quad
	\begin{matrix}
	\emph{MLE}  
	\end{matrix} \right)
	=P\left( \begin{matrix}
	\emph{non-negative } \\
	\emph{MDOOLSE} \\
	\end{matrix} \quad = \quad
	\begin{matrix}
	\emph{REMLE} 
	\end{matrix} \right)
	=1.$$
\end{theorem}
\begin{proof}
	See the Appendix.
\end{proof}

\paragraph{Computational tools for Gaussian orthogonal FDSLRM} 
\leavevmode\\[6pt]
The most important and useful consequence of Theorem \ref{thm:eqDOOLSEMLE} is that we can 
compute MLE or REMLE in orthogonal FDSLRM as quadratic NN-(M)DOOLSE. Therefore, we can use our KKT 
algorithm of order $\bigO(n)$ or computational tools CVXPY (CVXR) with $\bigO(n^3)$ described in 
the previous section.  Moreover, for $\epsv \in \LS{\V}$ NN-(M)DOOLSE do not fail as (RE)MLE, but 
naturally extend them.

Since a Gaussian orthogonal FDSLRM is also Gaussian LMM, we can also apply current LMM packages.
In our previous paper \cite{gajdos_kriging_2017}, we stated that no current package in R is 
directly and effectively suitable for FDSLRM. Thanks to a detailed study of  
\cite{demidenko_mixed_2013}, \cite{galecki_linear_2013} we were successfully directed and 
instructed how to implement the FDSLRM variance structure into one of the best R packages for LMM 
known 
as \textit{nlme} \label{abv-nlme} (\cite{pinheiro_mixed-effects_2009}, \cite{pinheiro_nlme:_2018}). 
After that, 
inspired by nlme,  we found another R package called \textit{sommer}\label{abv-sommer}, also 
suitable for FDSLRM 
fitting (\cite{covarrubias-pazaran_genome-assisted_2016}, \cite{cornuejols_optimization_2018}).

Simultaneously, thanks to online computing environment \textit{CoCalc} for SageMath 
\cite{stein_sage_2019} with the possibility to run computations and codes for free in many other 
programming languages and open softwares, we are also able to run and test MATLAB function 
\textit{mixed} (\cite{witkovsky_mixed_2018}, \cite{witkovsky_estimation_2012}) primarily intended 
for estimating variances in LMMs. On the basis of successful tests, we wrote its R version called 
\textit{MMEinR}\label{abv-MMEinR}, which is also available in our GitHub Repository 
\cite{gajdos_mmeinr:_2019}.

\begin{remark}
\small
\leavevmode\\
The mentioned LMM packages can also handle (RE)MLE in the general FDSLRM without orthogonality 
restriction. The packages use iterative methods based on EM algorithm or Henderson MME together 
with some version of the Newton-Raphson method whose complexity is generally at least $n^3$.

We are fully aware that another efficient implementation for estimating variances in LMMs provides 
\textit{lme4} package \cite{bates_fitting_2015} or SAS package \textit{PROC MIXED} 
\cite{stroup_sas_2018}, \cite{sasinstituteinc._sas/stat_2018}.  As for lme4, which is much faster 
than nlme, we found that the package does not allow implementing FDSLRM using lme4 standard input 
procedure. 

Finally, as SAS laymans, we also tried a university edition of SAS, free from 2014, but we left it 
after running into initial problems with Windows installation process in a virtual machine 
environment and subsequently with a more sophisticated programming language than Python or R. 
Therefore our knowledge with SAS remained only at the theoretical level.

\end{remark}

\section{Application in real data examples}

\subsection{Three real data sets and their FDSLRMs}

We illustrate the obtained theoretical results and the performance of the proposed EBLUP-NE method 
of variances $\nuv$  in the FDSLRM \eqref{def:FDSLRM} using three real data sets: 
\textit{electricity consumption}, \textit{tourism} and \textit{cyber security}. For all data sets, 
we identified the most parsimonious structure of the FDSLRM using an iterative process 
of the model building and selection based on exploratory tools of spectral analysis and LMM theory. 
Details of our analysis and modeling can be found in our easy reproducible Jupyter notebooks freely 
available at our GitHub 
repository   \cite{gajdos_fdslrm_2019}.

\paragraph{Electricity consumption}
\leavevmode\\[6pt]
As the first real data example, we have the econometric time series data set, representing 
hourly observations of the consumption of the electric energy (in kWh) in an department store. The 
number of time series observations is $n=24$. The data was adapted from 
\cite{stulajter_estimation_2004}. The consumption data can be fitted 
 by the following Gaussian 
orthogonal FDSLRM \cite{stulajter_estimation_2004}:
\begin{flalign} 
\qquad  X(t) = 
& \beta_1+\beta_2\cos(\omega_1 t)+\beta_3\sin(\omega_1 t)+ 
\nonumber\\
& +Y_1\cos(\omega_2 t)+Y_2\sin(\omega_2 t)+\label{model:electricity} \\
& + Y_3\cos(\omega_3 t)+Y_4\sin(\omega_3 t) 
+w(t), \, t\in \N & \nonumber
\nonumber
\end{flalign}
with $k=3, l=4, \;\Yv = (Y_1, Y_2, Y_3, Y_4)' \sim \NX_4(\Ov{4}, \D_{\nuv}),\; w(t) \sim iid \,
\NX(0,\s_0^2) $ and $\nuv= (\s_0^2, \s_1^2, \s_2^2, \s_3^2, \s_4^2)' \in \Rv{5}_{+}$. Frequencies 
$\omega_1,\omega_2,\omega_3 $ are suitable  Fourier frequencies from the time series periodogram  
in spectral analysis (\cite{priestley_spectral_2004}, \cite{stulajter_estimation_2004}, 
\cite{gajdos_kriging_2017}).

Since the time series data set contains only 24 observations and we have to estimate 
three regression and five variance parameters of the FDSLRM, this real data FDSLRM example should 
be considered only as a toy example. Taking two sets of Fourier frequencies 
$\left(2\pi/24,2\pi\cdot 3/24, 
2\pi \cdot 4/24\right)$ or $\left(2\pi/24,2\pi\cdot 2/24, 
2\pi \cdot 3/24\right)$, we get two toy models previously introduced in \cite{gajdos_kriging_2017} 
and \cite{stulajter_estimation_2004}. These models give us the opportunity to check our 
numerical results and to demonstrate how in principle FDSLRM estimation methods and computational 
tools work.

\paragraph{Tourism}
\leavevmode\\[6pt]
In this econometric FDSLRM application, we consider the time series data set, called 
\textit{visnights}, representing \textit{total quarterly visitor nights} (in millions) from 
1998-2016 in one of the regions of Australia -- inner zone of Victoria state. The number of time 
series observations is $n=76$. The data was adapted from \cite{hyndman_forecasting:_2018}.

The Gaussian orthogonal FDSLRM fitting the tourism data has the following form
(for details see our Jupyter notebook \texttt{tourism.ipynb}):
\begin{flalign}
\qquad  X(t) = 
& \beta_1+\beta_2\cos\left(\tfrac{2\pi t}{76}\right)+\beta_3\sin\left(\tfrac{2\pi t\cdot 
2}{76}\right)+ \label{model:tourism} \\
& +Y_1\cos\left(\tfrac{2\pi t\cdot 19 }{76}\right)+Y_2\sin\left(\tfrac{2\pi t\cdot 
19}{76}\right)+Y_3\cos\left(\tfrac{2\pi t\cdot 38}{76}\right)+w(t), \, t\in \N 
& \nonumber
\end{flalign}
with $k=3, l=3, \;\Yv = (Y_1, Y_2, Y_3)' \sim \mathcal{N}_3(\Ov{3}, \D_{\nuv})\,$  and
$\nuv= (\s_0^2, \s_1^2, \s_2^2, \s_3^2)' \in \Rv{4}_{+}$.

\paragraph{Cyber attacks}
\leavevmode\\[6pt]
Our final FDSLRM application describes the real time series data set representing the total weekly 
number of cyber attacks against a honeynet -- an unconventional tool which mimics real systems 
connected to Internet, like business or school computers intranets, to study methods, tools and 
goals of cyber attackers. Data, 
taken 
from \cite{sokol_prediction_2018}, were collected from November 2014 to May 2016 in CZ.NIC honeynet 
consisting of Kippo honeypots in medium-interaction mode. The number of time series observations is 
$n=72$. 

The suitable FDSLRM, after a preliminary logarithmic transformation of 
data $Z(t) = \log X(t)$, is again Gaussian orthogonal (for details see our 
Jupyter notebook \texttt{cyberattacks.ipynb}) and in comparison with 
previous models \eqref{model:electricity}, \eqref{model:tourism} has the 
simplest structure:
\begin{flalign}
\qquad  Z(t) = 
& \beta_1+\beta_2\cos\left(\tfrac{2\pi t}{72}\right)+\beta_3\sin\left(\tfrac{2\pi\cdot 3 
t}{72}\right)+ \label{model:security}\\
& +Y_1\cos\left(\tfrac{2\pi\cdot 14 t}{72}\right)+Y_2\sin\left(\tfrac{2\pi\cdot 14 
t}{72}\right)+w(t), \, t\in \N,
& \nonumber
\end{flalign}
with $k=3, l=2, \;\Yv = (Y_1, Y_2)' \sim \mathcal{N}_2(\Ov{2}, \D_{\nuv})\,$  and
$\nuv= (\s_0^2, \s_1^2, \s_2^2)' \in \Rv{3}_{+}$.

\subsection{Numerical results of estimating $\nuv$}

For cross-checking purposes, we realized our numerical computations in both Python and R based 
software tools (or packages). As we described in previous sections, we implemented own algorithms 
and methods in SciPy, SageMath and R, particularly analytical expressions \eqref{eq:orthoNE} for NE 
and the KKT algorithm for NN-(M)DOOLSE (tab.~\ref{tab:KKTalgorithm}) in orthogonal FDSLRM. 

Simultaneously, we confirmed the same estimations using CVXPY (or CVXR) package based on convex 
optimization and analogically using up-to-date standard LMM R packages nlme, MMEinR and sommer. 

Detailed computational results for all three data sets using all corresponding relevant tools can 
be found and reproduced in our collection of Jupyter notebooks (asterisk \texttt{*} represents a 
name of data and specific computational tool): \texttt{PY-estimation-*.ipynb} 
and \texttt{R-estimation-*.ipynb}.  

Table \ref{tab:estimates} summarizes all types of considered initial estimates and their 
corresponding EBLUP-NE in all four models. The computations also confirmed Theorem 
\ref{thm:eqDOOLSEMLE} therefore we wrote results for MLE and REMLE only.

\vspace{-15pt}

\begin{table}[H] 
	\renewcommand{\arraystretch}{1.25}
	\begin{center}
		\caption{Results for initial estimates and corresponding  EBLUP-NE. \newline Results for 
		(RE)MLE are identical with NN-(M)DOOLSE.}
		\label{tab:estimates}
		\begin{tabular}{c | c  c | c c}
			\toprule[1.5pt]
			
			\multicolumn{5}{c}{\textit{electricity consumption - toy model 1}} \\
			\midrule
			\textbf{method} & \textbf{estimate $\est\nuv$} & $\norm{\est\nuv}$ &\textbf{EBLUP-NE} 
			$\best\nuv$& $\norm{\best\nuv}$  \\
			\midrule
			NE    & $(3.53,0.37,1.86,0.00,1.26)'$ & $4.20$ & $(3.53,0.12,1.39,0.00,0.83)'$&$3.89$\\
			MLE   & $(2.86,0.13,1.62,0.00,1.03)'$ & $3.45$ & $(3.53,0.05,1.42,0.00,0.84)'$&$3.90$\\ 
			REMLE & $(3.34,0.09,1.59,0.00,0.99)'$ & $3.83$ & $(3.53,0.02,1.35,0.00,0.77)'$&$3.86$\\
			
			\midrule[1.5pt]
			\multicolumn{5}{c}{\textit{electricity consumption - toy model 2}} \\
			\midrule
			\textbf{method} & \textbf{estimate $\est\nuv$} & $\norm{\est\nuv}$ &\textbf{EBLUP-NE} 
			$\best\nuv$& $\norm{\best\nuv}$  \\
			\midrule
			NE    & $(1.09,2.97,1.76,0.37,1.86)'$ & $4.09$ & $(1.09,2.79,1.59,0.24,1.69)'$&$3.80$\\
			MLE   & $(0.93,2.89,1.68,0.29,1.79)'$ & $3.91$ & $(1.09,2.81,1.61,0.23,1.71)'$&$3.83$\\ 
			REMLE & $(1.09,2.87,1.67,0.28,1.77)'$ & $3.93$ & $(1.09,2.79,1.58,0.21,1.69)'$&$3.79$\\
			
			\midrule[1.5pt]
			\multicolumn{5}{c}{\textit{tourism}} \\
			\midrule
			NE    & $(0.108,0.004,0.230,0.022)'$ & $0.255$& $(0.108,0.001,0.225,0.020)'$ &$0.250$\\
			MLE	  & $(0.103,0.001,0.228,0.021)'$ & $0.251$& $(0.108,0.000,0.225,0.020)'$ &$0.250$\\ 
			REMLE & $(0.108,0.001,0.227,0.021)'$ & $0.253$& $(0.108,0.000,0.225,0.020)'$ &$0.250$\\
			
			\midrule[1.5pt]
			\multicolumn{5}{c}{\textit{cyber attacks}} \\
			\midrule
			NE    & $(0.0593,0.0255,0.0155)'$ & $0.0664$  & $(0.0593,0.0225,0.0127)'$  &$0.0647$\\
			MLE	  & $(0.0560,0.0239,0.0139)'$ & $0.0624$  & $(0.0593,0.0225,0.0125)'$  &$0.0647$ \\ 
			REMLE & $(0.0593,0.0238,0.0138)'$ & $0.0654$  & $(0.0593,0.0223,0.0124)'$  &$0.0646$\\
			
			\bottomrule[1.5pt]
			
		\end{tabular} 
	\end{center}
\end{table}

\vspace{-15pt}

Finally, we point out that thanks to SageMath and our very fast KKT algorithm, we were able 
to compute (in real time) NN-(M)-DOOLSE in toy examples with infinite precision -- as the exact 
closed-form (algebraic) numbers. To get an explicit idea, e.g. for our first toy model, we got 
these results for NN-MDOOLSE $\est{\nuv}_+$ identical with REMLE

$$
\renewcommand{\arraystretch}{1.4}
\est{\nuv}_{+} =
\begin{pmatrix}
\s_0^2 \\
\s_1^2 \\
\s_2^2 \\
\s_3^2 \\
\s_4^2
\end{pmatrix}
=
\begin{pmatrix}
-\frac{6569}{4320} \, \sqrt{3} \sqrt{2} - \frac{46513}{21600} \, \sqrt{3} - \frac{7511}{2400} \, 
\sqrt{2} + 
\frac{328739}{21600} \\
\frac{6569}{51840} \, \sqrt{3} \sqrt{2} + \frac{46513}{259200} \, \sqrt{3} + \frac{11291}{28800} \, 
\sqrt{2} - 
\frac{56089}{51840} \\
\frac{6569}{51840} \, \sqrt{3} \sqrt{2} + \frac{46513}{259200} \, \sqrt{3} + \frac{2803}{3200} \, 
\sqrt{2} - 
\frac{71213}{259200} \\
0 \\
\frac{6569}{51840} \, \sqrt{3} \sqrt{2} + \frac{46513}{259200} \, \sqrt{3} + \frac{7511}{28800} \, 
\sqrt{2} - 
\frac{203}{259200}
\end{pmatrix} \approx
\begin{pmatrix}
3.339 \\
0.094 \\
1.586 \\
0 \\
0.989
\end{pmatrix}.
$$
It means that we can compute \textit{real} errors in results for all used computational tools to 
explore their quality from the viewpoint of numerical precision. At the same time, we also watched 
a run time for particular computational tools using Jupyter notebook extension ExecuteTime, which 
provided a preliminary comparison of real execution times. 

\section{Conclusions}

We suggested and investigated an alternative, new method EBLUP-NE based on empirical BLUPs for 
estimating variances in time series modeled by FDSLRM. The estimation method  can be also 
viewed, analogously like EBLUP (\cite{rao_small_2015}, \cite{witkovsky_estimation_2012}), as a 
two-stage iterative method with one step in the iteration. 

EBLUP-NE are invariant quadratic, non-negative estimators whose simple computational form and 
subsequent first and second-moment statistical properties are given by two special matrices -- the 
Schur complement and Gram matrix \cite{zhang_schur_2005}. The method can be used not only in the 
case of normally distributed time series data, but for any absolutely continuous probability 
distribution of time series data. 

As initial starting estimates for EBLUP-NE, we can principally employ any of the previously used 
methods based on least squares (NE, NN-DOOLSE, NN-MDOOLSE) or maximum likelihood (MLE, REMLE). 
Applications of FDSLRM with the EBLUP-NE on three real data sets (electricity consumption, tourism, 
cyber attacks) providing two toy and two real models indicate that the method computationally gives 
at least comparable results with REMLE or NN-MDOOLSE, but in faster run time (approximately 10-1000 
times on the standard PC).

Due to lack of computational implementation for FDSLRM modeling, which would be generally available 
and readily applicable, and the fact that FDSLRM least squares and maximum likelihoods estimation 
procedures are more than 10 years old, we revisited and updated theoretical and computational 
knowledge dealing with these methods. Specifically, applying the convex optimization theory, we 
reformulated all estimation methods as convex optimization problems in the so-called orthogonal 
FDSLRM, the most usual form of FDSLRM used in practice. 

We formulated the KKT optimality conditions, which, unlike likelihood 
equations, are necessary and sufficient conditions for optimal solutions of (RE)MLE or NN-(M)DOOLSE 
on extended parametric space $\Ys = [ 0, \infty)^{l+1}$.
KKT optimality conditions dictate not only the exact existence conditions of estimates but they 
also solve the well-known problem dealing with standardly used likelihood 
equations \cite[chap. 12]{christensen_plane_2011}, \cite[chap. 
6]{searle_variance_2009} in LMM where their solutions for (RE)MLE or NN-(M)DOOLSE may not be 
required estimates --- they may be out of the parameter space or they may be other than the 
maximum or minimum. 

Moreover, using KKT optimality conditions in orthogonal FDSLRM, we proved the equivalence of 
NN-(M)DOOLSE and (RE)MLE with propability 1. This most important theoretical result of the paper is 
the stronger and more general result than in \cite{stulajter_estimation_2004} proved only for 
interior points of $\Ys$. 

Simultaneously, the convex optimization theory brought us to the new KKT algorithm for computing 
NN-(M)DOOLSE, equivalent to (RE)MLE, with double floating-point precision $\epsilon<10^{-15}$ as 
the 
default 
precision of outputs 
and with computational complexity $\bigO(n)$. Such an algorithm which we implemented in SciPy, 
SageMath and R, which at the default precision level is $10^7$ times more accurate and 
approximately $n^2$ faster than the best current Python- or R-based standard computational tools, 
can be used in effective computational time series research to study properties of FDSLRM (Monte 
Carlo and bootstrap methods, \cite{kreiss_bootstrap_2012}).
 
Regarding computational aspects, we were also successful in the identification and demonstration of 
consistent results for the real data applications of FDSLRM in several free, open-source current 
standard computational tools --- namely CVXPY, CVXR (R version of CVXPY) packages for convex 
optimization and LMM R packages nlme, sommer and MMEinR (our R version of Witkovsky's MATLAB mixed 
function). These results and procedures can be freely viewed in our 16 Jupyter notebooks which are 
easily readable, sharable, reproducible and modifiable directly in our GitHub repository 
(\cite{gajdos_fdslrm_2019}). Open-source Jupyter technology with Python and R packages also 
solved our problem stated in \cite{gajdos_kriging_2017} that no current package in R was directly 
and effectively suitable for FDSLRM. 
  
Finally, our investigation has also brought new questions for further 
research. There is definitely a need for more exact and detailed analysis based on a simulation 
study and general EBLUP theory (\cite{zadlo_mse_2009}, \cite{witkovsky_estimation_2012}, 
\cite{rao_small_2015}) focusing on EBLUP-NE quality with respect to previously used estimation 
methods and the performance of the EBLUP-NE method itself with different initial starting points, 
using different computational tools in various probability distributions. These research questions 
are under our investigation and will be published in the near future.   
  
In connection with our current computational research in FDSLRM, but also for real time series data 
analysis and forecasting, we started to build our own R package (see a preliminary, fully 
functional version at \cite{gajdos_fdslrmR_2019}) on mentioned LMM R packages to manipulate readily 
with FDSLRM concepts and procedures.

Finally, our results in the paper can be seen reciprocally as contributions to convex 
optimization and LMM methodology. Particularly, our convex optimization application in the context 
of time series modeling has become another one from a wide variety of application areas of convex 
optimization \cite{boyd_convex_2009}. Since FDSLRM describing $n$ observed time series values is 
also a special type of LMM, our EBLUP-NE and the very fast, accurate KKT algorithm to compute 
(RE)MLE may have potential to be used in computational research and applications dealing with LMMs. 

\vspace{-6pt}

\begin{acknowledgements} 
\leavevmode\\
We are very grateful especially to Franti\v{s}ek \v{S}tulajter (Comenius University in Bratislava, 
SK) for his guidance during our early contact and study of FDSLRM concepts and to 
Viktor Witkovsk\'y (Slovak Academy of Science, SK) for his recommendations and deep insights 
dealing with the LMM methodology and its computational tools connected to FDSLRM. We would also 
like to thank Ivan \v{Z}ezula and Daniel Klein (both from P.~J.~\v{S}af\'arik University in 
Ko\v{s}ice, SK) for a valuable feedback dealing with considered estimation methods. 

Concerning applied computational tools, we would also like to acknowledge involvement and 
recommendations of Giovanny Covarrubias-Pazaran (University of Wisconsin, US) in using sommer 
package, Steve Diamond and Stephen Boyd (Stanford University, US) in using CVXPY and CVXR. Finally, 
in connection with using SageMath, Jupyter and GitHub, our specials thanks go namely to William 
Stein (University of Washington, US), Harald Schilly (Universit\"{a}t Wien, AT), Erik Bray 
(Université Paris-Sud, FR), Luca de Feo (Université de Versailles Saint-Quentin-en-Yvelines, 
FR) and Charles J. Weiss (Augustana University, US).
\end{acknowledgements}

\vspace{-10pt}

\section*{Appendix}

\small

\paragraph{Acronyms and abbreviations}
\leavevmode\\

\vspace{-20pt}

\begin{table}[H]
	\renewcommand{\arraystretch}{1.25}
	\begin{center}
	\caption{List of acronyms used in the paper.}
	\label{tab:acronyms}
\begin{tabular}{l p{6cm} c}
\toprule[1.5pt]
\textbf{acronym} & \textbf{explanation} & \textbf{specification}\\
\midrule
\small
BLUE      & best linear unbiased estimator         & \eqref{eq:MME}, \pref{eq:MME} \\
BLUP      & best linear unbiased predictor         & \eqref{eq:MME}, \pref{eq:MME} \\
BLUP-NE   & natural estimators based on BLUP       & \eqref{def:BLUPNE}, \pref{def:BLUPNE} \\
CVXPY     & Python library for convex optimization & \pref{abv-CVXPY} \\
CVXR      & R version of CVXPY                     & \pref{abv-CVXPY} \\
DOOLSE    & double ordinary least squares estimator\newline
            without non-negativity constraints     & \eqref{eq:projectionMDOOLSE},  
                                                     \pref{eq:projectionMDOOLSE} \\
EBLUP     & empirical (plug-in) BLUP               & \pref{abv-EBLUP} \\
EBLUP-NE  & natural estimators based on EBLUP      & \pref{abv-EBLUPNE} \\
FDSLRM    & finite discrete spectrum \newline 
            linear regression model                & \eqref{def:FDSLRM}, 
\pref{def:FDSLRM} \\
KKT       & Karush-Kuhn-Tucker                     & \pref{abv-KKT} \\
LMM       & linear mixed model                     & \pref{eq:FDSLRMobservation} \\
LRM       & linear regression model                & \pref{abv-LRM} \\
MDOOLSE   & modified (unbiased) DOOLSE \newline 
without non-negativity constraints                 & \eqref{eq:projectionMDOOLSE},  
												     \pref{eq:projectionMDOOLSE} \\
(M)DOOLSE & considering both DOOLSE and MDOOLSE    & Remark \ref{rem:(M)DOOLSE},  
							                         \pref{rem:(M)DOOLSE} \\
MLE 	  & maximum likelihood estimators          & \eqref{def:MLE}, \pref{def:MLE} \\
MME 	  & Henderson's mixed model equations      & \eqref{eq:MME}, \pref{eq:MME}   \\
MMEinR 	  & R version of MATLAB function mixed     & \pref{abv-MMEinR} \\
MSE       & mean squared error 		               & \pref{abv-MSE} \\
NE        & natural estimators 					   & \pref{abv-NE}, \eqref{eq:orthoNE}, 
													 \pref{eq:orthoNE} \\
nlme      & R package for (non) linear \newline 
			mixed(-effects) models 				   & \pref{abv-nlme} \\
NN-DOOLSE & non-negative DOOLSE                    & \eqref{def:DOOLSE}, \pref{def:DOOLSE} \\
NN-MDOOLSE& non-negative modified DOOLSE           & \eqref{def:MDOOLSE}, \pref{def:MDOOLSE}\\
NN-(M)DOOLSE& considering both NN-DOOLSE \newline 
              and NN-MDOOLSE                       & Remark \ref{rem:(M)DOOLSE},  
                                                     \pref{rem:(M)DOOLSE}  \\
OLS       & ordinary least squares                 & \pref{abv-OLS} \\
REMLE     & residual (restricted) \newline 
            maximum likelihood estimator           & \eqref{def:REMLE}, \pref{def:REMLE} \\
(RE)MLE   & considering both MLE and REMLE         & Remark \ref{rem:(M)DOOLSE},  
                                                     \pref{rem:(M)DOOLSE} \\
SageMath  & free Python-based mathematics software & \pref{abv-SageMath} \\
sommer    & R package for multivariate LMMs  \newline
            solving MME                            & \pref{abv-sommer} \\
SciPy     & Scientific Python, \newline 
            Python-based ecosystem of open software& \pref{abv-SciPy}\\
\bottomrule[1.5pt]
\end{tabular}
	\end{center}
\end{table}

\newpage

\paragraph{The BLUP-NE method}

\begin{proof}[Theorem \ref{thm:propertiesBLUPNE}] \\
Employing Lemma \ref{thm:propertiesT} and the well-known standard expressions for
mean values and covariances of invariant quadratic estimators (see e.g. 
\cite{christensen_plane_2011}, \cite{puntanen_formulas_2013})) 
$\En{\Xv'\A \Xv}=\tr(\A\Sgn)$ and if  
$\Xv \sim  \NnX$ then $\Covn{\Xv'\A \Xv, \Xv'\B \Xv}=2\,\tr(\A\Sgn \B\Sgn)$, we have

\begin{enumeratep}
	
\item $\En{\bess{j}}=E(\Xv\dt_j\dt_j'\Xv)=\tr(\dt_j\dt_j'\Sgn)=
	\tr(\dt_j'\Sgn \dt_j)=
	(\dT\Sgn \dT')_{jj}$. \\
	According to (3) of Lema \ref{thm:propertiesT} $(\dT\Sgn \dT')_{jj}= 
	\D_{jj}-\s^2(\dW^{-1})_{jj}=
	\s^2_j-\s^2(\dW^{-1})_{jj}$.
	
\item 
	$\Din{\bess{j}}=
	2\tr(\dt_j\dt_j'\Sgn \dt_j\dt_j' \Sgn)=
	2\tr(\dt_j' \Sgn \dt_j' \dt_j \Sgn \dt_j)=
	2(\dT\Sgn\dT')_{jj}$ \\ As a consequence of (1)
	$\Din{\bess{j}} =2(\s^2_j-\s^2(\dW^{-1})_{jj})^ 2$.
	
\item  In a similar way as in (2)
	
	$\Covn{\bess{i},\bess{j}}=
	2\tr(\dt_i\dt_i'\Sgn \dt_j \dt_j' \Sgn)=
	2\tr(\dt_i'\Sgn \dt_j \dt_j' \Sgn \dt_i)=
	2(\dt_i'\Sgn \dt_j)^2= \\
	=2(\dT \Sgn \dT')^2_{ij}=
	2(0-\s^2 (\dW^{-1})_{ij})^2=
	2(\s^2 (\dW^{-1})_{ij})^2$.
	
\item 
	$\MSEn{\bess{j}}=\En{{(\bess{j}-\s^2_j)}^2}=
	{(\En{\bess{j}}-\s^2_j)}^2+\Din{\bess{j}}=\\
	=(\s^2_j-\s^2(\dW^{-1})_{jj}-\s^2_j)^2+2(\s^2_j-\s^2(\dW^{-1})_{jj})^2=\\
	=\s^4(\dW^{-1})^2_{jj}+2(\s^2_j-\s^2(\dW^{-1})_{jj})^2$.
\end{enumeratep} 
\end{proof}\smartqed

\paragraph{Existence of inversions for $ \dU$,  $\dHV $}
\leavevmode\\[6pt]
To prove nonsingularity of $ \dU$,  $\dHV $, it is sufficient to show that both matrices have 
non-zero determinants. Using idempotence of $\Mf$ and expression  \cite[sec. 
6.8]{searle_matrix_2017}
\begin{equation} \label{eq:propertydet}
\det (\lambda \In-\A\B) = 
\lambda^{n-l}\det (\lambda \Il-\B\A) \text{ for } \lambda \neq 0, \A \in \Rm{n}{l}, \B \in 
\Rm{l}{n},
\end{equation}
we can write for determinants of $ \dU$,  $\dHV $ 
\begin{equation}
\begin{aligned}
\det \dHV  = \det (\s_0^2 \Il+\V'\V\D) &= (-\s_0^2)^{l-n}\det (\s_0^2 \Il+\V\D\V') \\
\det \dU = \det (\s_0^2 \Il+\V'\Mf\Mf\V\D)& = (-\s_0^2)^{l-n}\det (\s_0^2 \Il+\Mf\V\D\V'\Mf) 
\end{aligned}
\end{equation} 
Now we can see that the sum of positive definite matrix $\s_0^2\Il \,\, (\s_0>0)$ and each of 
positive semidefinite matrices $\V\D\V', \,\,\Mf\V\D\V'\Mf$ is a positive definite matrix, whose 
determinant is always positive.

\paragraph{Double least squares estimators NN-(M)DOOLSE}
\leavevmode\\[6pt]
Using the standard inner product of matrices defined as $(\bullet, \bullet) = \tr(\bullet, 
\bullet)$,  generating the Euclidean norm $\norm{\bullet}$ of a matrix, and basic properties of the 
trace function,  we can easily rewrite expressions \eqref{def:DOOLSE}, \eqref{def:MDOOLSE} for the 
objective functions as quadratic forms. The particular form of these quadratic forms in 
optimization problems (M)DOOLSE is described by the following proposition.
\begin{proposition}[NN-(M)DOOLSE as quadratic optimization problems]\label{pro:cvxMDOOLSE}
	\leavevmode\\
	In an orthogonal FDSLRM the \emph{NN-(M)DOOLSE} problems are convex quadratic optimization 
	problems in 
	the form
	\begin{flalign} 
	\qquad&
	\begin{array}{ll} 
	\emph{minimize}  & f_0(\nuv)= \nuv'\G\nuv - 2\qv'\nuv\\[6pt]
	\emph{subject to} & -\I{l+1}\nuv \preceq \Ov{l+1}
	\end{array}  \label{eq:MDOOLSEquadratic}
	&&
	\end{flalign}
	$\qquad \qv = 
	\left(\begin{array}{c}
	\epsv'  \epsv\\
	(\epsv'  \vv_{1})^2 \\
	(\epsv'  \vv_{2})^2 \\
	\vdots \\
	(\epsv'  \vv_{l})^2
	\end{array}\right)$,
	$ \quad \G = \left(\begin{array}{ccccc}
	n^* & \norm{\vv_{1}}^2 & \norm{\vv_{2}}^2 & \ldots & 
	\norm{\vv_{l}}^2 \\
	\norm{\vv_{1}}^2 & \norm{\vv_{1}}^4 & 0 & \ldots & 0 \\
	\norm{\vv_{2}}^2 & 0 & \norm{\vv_{2}}^4 & \ldots & 0 \\
	\vdots & \vdots & \vdots & \ldots & \vdots \\
	\norm{\vv_{l}}^2 & 0 & 0 & \ldots & \norm{\vv_{l}}^4
	\end{array}\right)\succ0
	$ \\[12pt]
	where for \emph{NN-DOOLSE: } $n^* \equiv n$, \,\emph{NN-MDOOLSE: }$n^* \equiv n-k$,
\end{proposition}

According to the fundamental KKT theorem of convex optimization \cite[chap.5]{boyd_convex_2009} 
which handles convex optimization problems with constraints in the form of inequalities giving a 
list of the so-called \textit{Karush-Kuhn-Tucker (KKT) optimality conditions}, we consider our 
optimization 
problem NN-(M)DOOLSE in the Lagrange multiplier form with Lagrangian $L$ as a sum of the objective 
function $f_0(\nuv)$ and a linear combination of multipliers $\lambdav = (\lambdas_0, \lambdas_1, 
\ldots, \lambdas_l)' \succeq \Ov{l+1}$ and constraints $-\nuv =(\nus_0, \nus_1, \ldots, \nus_l)'
\preceq \Ov{l+1}$. 

Then KKT optimality conditions for \eqref{eq:MDOOLSEquadratic} become necessary and sufficient 
conditions of optimality, satisfied at any local optimal solution, being represented by three sets 
of conditions: (1) primal and dual 
feasibility: $\nuv \succeq \Ov{l+1}$, $\lambda\succeq \Ov{l+1}$, (2) stationarity of the 
Lagrangian: $\nabla_{\nuv} L(\nuv, \lambdav) = \Ov{l+1}$ and (3) complementary slackness:  
$\nus_j\lambdas_j=0,\, j\in\{0,1, \ldots,l\}$. Computing the gradient of the Lagrangian, we arrive  
at the following proposition. 

\begin{proposition}[KKT optimality conditions for NN-(M)DOOLSE]\label{pro:KKTcvxMDOOLSE}
	\leavevmode \\
	In an orthogonal FDSLRM consider the non-negative \emph{NN-(M)DOOLSE} problems in~the Lagrange 
	multiplier form with the Lagrangian
	$$L(\nuv, \lambdav) =f_0(\nuv)-\SumZ{j}{l}\lambdas_j\nus_j.$$
	Then, a \textbf{necessary and sufficient condition} for $\est\nuv, \est\lambdav$ to be 
	problems' optimal solution is
	\begin{enumeratep}
		\item  $\nuv \succeq \Ov{l+1}, \quad \lambdav \succeq \Ov{l+1}$
		\item  $\begin{pmatrix}
		\G & -\I{l+1} \\
		\end{pmatrix}
		\begin{pmatrix}
		\nuv \\
		\lambdav
		\end{pmatrix}
		= \qv$
		\item  $\nuv \circ \lambdav = \Ov{l+1} \quad \left( \Leftrightarrow \nus_j\lambdas_j=0,\,\, 
		j\in\{0,1, \ldots,l\}\right)$.
	\end{enumeratep}
\end{proposition}

As for the existence of optimal solutions, we can apply the well-known basic results for minimizing 
quadratic forms \cite[chap.~3]{bertsekas_convex_2009}. Since the Hessian of 
\eqref{eq:MDOOLSEquadratic} equal to $\nabla_{\nuv}^2 f_0(\nuv) = 2\G$ is a positive 
definite matrix, like $\G$, then $f_0(\nuv)$ is strictly convex and also coercive proper function.
These two properties are sufficient conditions for NN-(M)DOOLSE \cite[Weierstrass' theorem, 
p.~119]{bertsekas_convex_2009} to have always a unique global optimal 
solution $\est\nuv$.

Employing the familiar Bessel's inequality $\Sum{j}{l}(\epsv'\vv_j/\norm{\vv_j})^2 \leq 
\norm{\epsv}^2$, where equality holds if and only if $\epsv \in \LS{\V}$, it is also not difficult 
to see that KKT conditions (1-3) rewritten as the following system
\vspace{-12pt}
\begin{align}
n^*\nus_0 + \Sum{j}{l} \nus_j \norm{\vv_j}^2 - \lambdas_0 & = \norm{\epsv}^2 \label{eq:KKT0}\\
\nus_0 \norm{\vv_j}^2 + \nus_j \norm{\vv_j}^4 - \lambdas_j & = (\epsv'\vv_j)^2; \, \,j = 1,\ldots, 
l \label{eq:KKTj}\\[6pt]
\nus_j\lambdas_j =0, \,\,  \lambdas_j \geq   0, & \, \,  \nus_j  \geq 0; \, \, j = 0,1,\ldots, l 
\label{eq:slack}
\end{align}
imply an optimal solution $\est\nuv$ with $\est\nus_0=0$ if and only if the vector of OLS residuals 
$\epsv$ belongs to the column space of $\V$. Since probability of $\epsv \in \LS{\V}$ is zero, our 
existence conclusions can be summarized in the next proposition.

\begin{proposition}[Existence of NN-(M)DOOLSE]\label{pro:cvxMDOOLSEexist}
	\leavevmode \\
	In an orthogonal FDSLRM the following holds
	\begin{enumerate}[label=\pitem]
		\item \emph{NN-(M)DOOLSE} problems are strictly quadratic optimization problems.
		\item Their solutions $\est\nuv$ always exist and they are unique global minimizers.
		\item $\est\nus_0=0 \Leftrightarrow \epsv \in \LS\V \Leftrightarrow  \xv \in \LS{\F, \V} $.
		\item $P(\est\nus_0=0)=P(\epsv \in \LS\V)=P\left(\xv\in\LS{\F, \V}\right)=0$.
	\end{enumerate}
\end{proposition}

\paragraph{KKT algorithm for NN-(M)DOOLSE}
\leavevmode\\
If we consider $\nus_0 = 0$ which occurs if and only if $\epsv = \Sum{j}{l} \alpha_j 
\vv_j$ in KKT conditions \eqref{eq:KKT0}-\eqref{eq:slack} then the optimal solution $\est\nuv_+$ 
for NN-(M)DOOLSE is trivial: $\nus_{+0} = 0, \nus_{+j} = \alpha_j^2, j =1, \ldots, l$. 

For  $\nus_0 \neq 0$ implying 
$\lambdas_0 =0$, $\nus_j$ and $\lambdas_j$ cannot be simultaneous zero otherwise it would lead to 
contradiction between \eqref{eq:KKT0} and \eqref{eq:KKTj}. 
Therefore, in the case $\nus_0 \neq 0$, the complementary slackness condition (3) 
$\nus_j\lambdas_j=0,\,\, j\in\{1, \ldots,l\} 
$ can be rewritten in the form $\,b_j \nus_j (1-b_j) \lambdas_j = 0 $, where $b_j$ is an auxiliary 
indicator, which is zero if $\nus_j = 0, \lambdas_j \ne 0 $ and one if $\nus_j \ne 0, \lambdas_j = 
0 
$. 

Using vector $\bvv = (b_1, b_2, \ldots, b_l)' \in 
\{0,1\}^l$, the derived KKT conditions (2) in Proposition \ref{pro:KKTcvxMDOOLSE} can be described 
by a $(l+1) \times (l+1)$ matrix function $\Kb = \{K_{ij}(\bvv)\}$ and a vector function $\gv(\bvv) 
= (\gs_0(\bvv),\gs_1(\bvv), \ldots, \gs_l(\bvv))'$  as
\begin{equation} \label{eq:Kbgb}
\Kb \gv(\bvv) = \qv,
\end{equation}
where
$$
K_{ij}(\bvv) =\left\{
\begin{array}{rl}
0, & \text{ if } i=0, b_j = 0,\\
-1, & \text{ if } i=j\ne 0, b_j = 0, \\
G_ {ij}, & \text{ otherwise.}
\end{array}
\right. \quad
\gs_j(\bvv) = \left\{
\begin{array}{cl}
\nus_0, & \text{ if } j = 0\\
b_j \nus_j +(1-b_j) \lambdas_j, & \text{ otherwise.}
\end{array}
\right.
$$
Applying the Banachiewicz formula, we can write for the inverse of $\Kb$ the 
following analytic expression
$$K(\bvv)^{-1} = \phi^{-1}
\begin{pmatrix}
1 & -\bvv'\GV \invDb \\
-\invDb \GV \jv{l} & \quad \phi\invDb+\invDb \GV \jv{l} \bvv'\GV\invDb
\end{pmatrix}$$
\noindent where \\
\indent$\Db = \diag{b_j \norm{\vv_j}^4+b_j-1}$, $\GV = \diag{\norm{\vv_j}^2}$, \\
\indent$\jv{l} = (1, 1, \ldots, 1)' \in \Rv{l}$, $\phi = n^{*} -\bvv'\GV \invDb \GV \jv{l}$. \\

\noindent Finally, Proposition \ref{pro:cvxMDOOLSEexist} guarantees the existence of unique 
auxiliary vectors $\dest{\gv}{},\dest{\bvv}{} $:
$$\dest{\gv}{}=\{\gv; \gv=\Kb^{-1}\qv\geqq 0, \bvv \in \{0,1\}^l\}$$
$$\dest{\bvv}{}=\{\bvv; \dest{\gv}{}=\Kb^{-1}\qv, \bvv \in \{0,1\}^l\}$$\\
\noindent  Based on vectors $\dest{\gv}{},\dest{\bvv}{} $, the NN-(M)DOOLSE $\nndest{\nuv}= 
(\dest{\nus}{+0}, \dest{\nus}{+1}, \ldots, \dest{\nus}{+l})'$ of $\nuv$ as a solution of KKT 
conditions has the final form
\begin{equation*}
\dest{\nus}{+j} =\left\{
\begin{array}{rl}
0, & \text{ if } j\ne 0, b_j = 0, \\
\dest{\gs}{j}, & \text{ otherwise.}
\end{array}
\right.
\end{equation*}
Thanks to Proposition 4, it is worth to mention that the matrix system \eqref{eq:Kbgb} includes 
also 
solutions with $\nus_0 = 0$.

\paragraph{Maximum likelihood estimators (RE)MLE}
\leavevmode\\[6pt]
Generally, (RE)MLE in FDSLRM are not convex optimization problems with respect to $\nuv$ and they 
do not exist for $\nus_0 = 0$ which occurs if and only if $\epsv \in \LS{V}$. If we apply the 
following 
bijective transformation defined on $\Ys = (0,\infty)\times [0, 
\infty)^l$ 
\begin{equation} \label{def:bijection}
d_0=\tfrac{1}{\nus_0},\,\,d_j=\tfrac{\nus_j}{\nus_0(\nus_0+\norm{\vv_j}^2\nus_j)},\,\,j\in\{1,
\ldots,l\},
\end{equation}
we can convert the (RE)MLE 
problems in the form of equivalent convex problems whose solutions can be readily converted to a 
(RE)MLE solutions by the inverse transformation to \eqref{def:bijection}. Here are more detailed 
steps of the conversion. 

In orthogonal FDSLRM $(\F'\V =0, \V'\V = \diag{\norm{\vv_j}^2})$, $\InvSgn$ \cite[lem 
2.1]{stulajter_mse_2003} and $\epsv' \InvSgn \epsv$ are equal to 
\begin{equation} \label{eq:InvSigma}
\begin{gathered}
\InvSgn = \tfrac{1}{\nus_0} \In - 
\Sum{j}{l}\tfrac{\nus_j}{\nus_0(\nus_0+\norm{\vv_j}^2\nus_j)}\vv_j\vv_j' = 
d_0 \In - \V \diag{d_j} \V', \\
\gnorm{\epsv}{\InvSgn}^2 = \epsv'\InvSgn\epsv = d_0\epsv'\epsv-\epsv'\V\,\diag{d_j}\V'\epsv. 
\end{gathered}
\end{equation}

Using orthogonality conditions \eqref{eq:orthogonality} and expression \eqref{eq:propertydet}, we 
get for determinants in 
(RE)ML loglikelihoods \eqref{def:MLE}, 
\eqref{def:REMLE}
\begin{equation} \label{eq:determinants}
\begin{gathered} 
\det(\F'\InvSgn \F) = \det (d_0 \F'\F) = d_0^k \det(\F'\F), \\
\det \InvSgn = d_0^{n-l}\det(d_0 \Il - \V'\V \diag{d_j})= d_0^{n-l}\prod_{j = 1}^{l}(d_0-d_j 
\norm{\vv_j}^2).
\end{gathered}
\end{equation}

Since the chosen bijective transformation \eqref{def:bijection} also transforms convex constraints 
for $\nuv$ to convex constraints for $\dvv$, substituting \eqref{eq:InvSigma}, 
\eqref{eq:determinants} into objective functions  \eqref{def:MLE}, 
\eqref{def:REMLE} of (RE)MLE, we  can formulate the following proposition.

\begin{proposition}[Equivalent (RE)MLE convex problems]\label{pro:cvxREMLE}
	\leavevmode\\
	Let assume  $\epsv \not\in \LS{\V}$ and consider a \textit{bijective transformation} in the 
	following form:
	$$d_0=\dfrac{1}{\nus_0},\,\,d_j=\dfrac{\nus_j}{\nus_0(\nus_0+\norm{\vv_j}^2\nus_j)},\,\,j\in\{1,
	\ldots,l\}.$$ 	
	Then, in a Gaussian orthogonal FDSLRM the \emph{(RE)MLE} problems are equivalent to convex 
	problems:
	\begin{flalign}
	&
	\begin{array}{ll} 
	\emph{minimize}  & {\small\quad f_0(\dvv)=-(n^*\!-l)\ln d_0 - \displaystyle\Sum{j}{l} 
		\ln(d_0-d_j\norm{\vv_j}^2)  
		+d_0\epsv'\epsv-\epsv'\V\,\diag{d_j}\V'\epsv } \\[6pt]
	\emph{subject to}& \quad d_0 > \max\{d_j\norm{\vv_j}^2, j = 1, \ldots, l\} \\
	& \quad d_j \geq 0, j=1,\ldots l
	\end{array} 
	&&  \nonumber
	\end{flalign}
	where for $\emph{MLE: } n^* \equiv n,$ and  $ \emph{for REMLE: } n^* \equiv n - k $.
\end{proposition}

Once again, we can write the equivalent (RE)MLE problem  with $2l$ constraints in the corresponding 
Lagrangian multiplier form based on the following Lagrangian with $2l$ multipliers $(\lambda_1, 
\ldots, 
\lambda_l, \mu_1, \ldots, \mu_l)' = (\lambdav', \muv')' \in \Rv{2l}_+$ 
\vspace{-6pt}
$$L(\dvv,\lambdav,\boldsymbol{\mu})=f_0(\dvv)-\Sum{j}{l}\left[\lambda_jd_j+\mu_j(d_0-d_j||\vv_j||^2)\right].$$
Then by the direct computation, we obtain KKT conditions for equivalent (RE)MLE describing (1) 
primal 
and dual feasibility for $\dvv$ and $(\lambdav', \muv')'$, (2) stationarity of the Lagrangian 
($\nabla_{\dvv}L(\dvv, \lambdav, \muv) = 0$)  and (3) complementary slackness, all described by the 
next proposition.  

\begin{proposition}[KKT optimality conditions for equivalent (RE)MLE]\label{pro:KKTcvxREMLE}
	\leavevmode\\
	Consider equivalent convex optimization problems to \emph{(RE)MLE} in the Lagrangian multiplier 
	form  
	\vspace{-6pt}
	$$L(\dvv,\lambdav,\muv)=f_0(\dvv)-\Sum{j}{l}\left[\lambda_jd_j+\mu_j(d_0-d_j||\vv_j||^2)\right].$$
	Then, a \textbf{necessary and sufficient condition} for 	
	$\hat{\dvv},\hat{\lambdav},\hat{\muv}$	to be problems' optimal solution is 
		\begin{enumeratep}
	\vspace{6pt}
	\item 
	$d_0-d_j\norm{\vv_j}^2>0,\,\, d_j \geq 0, \,\,\lambda_j\ge0,\,\,\mu_j\ge0,\,\, 
	j\in\{1,\ldots,l\}$
		\item  
		$||\epsv||^2-\dfrac{n^{*}-l}{d_0}-\displaystyle\Sum{j}{l}\left(\mu_j+
		\dfrac{1}{d_0-d_j||\vv_j||^2}\right)=0, $		\\
		$\qquad \dfrac{||\vv_j||^2}{d_0-d_j||\vv_j||^2}-(\epsv'\vv_j)^2-
		\lambda_j+\mu_j||\vv_j||^2=0,\,\,j\in\{1,\ldots,l\}$ 
		\item 
		\vspace{6pt}
		$-d_j\lambda_j=0,\,\,-(d_0-d_j||\vv_j||^2)\mu_j=0,
		j\in\{1,\ldots,l\}$.
	\end{enumeratep}
\end{proposition}

\vspace{12pt}

Similarly as in the case of NN-(M)-DOOLSE, the Hessian  $\Hs = \nabla_{\dvv}^2 f_0(\dvv)$ equal to 

$$ 
\renewcommand{\arraystretch}{2}
\begin{pmatrix}
\tfrac{n^* - l}{d_0^2} + \Sum{j}{l} \tfrac1{{(d_j \norm{\vv_j}^2 - d_0)}^2} & 
-\tfrac{\norm{\vv_1}^2}{{(d_1 \norm{\vv_1}^2 - d_0)}^2} & -\tfrac{\norm{\vv_2}^2}{{(d_2 
\norm{\vv_2}^2 - d_0)}^2} & \ldots & 
-\tfrac{\norm{\vv_l}^2}{{(d_l \norm{\vv_l}^2 - d_0)}^2} \\
-\tfrac{\norm{\vv_1}^2}{{(d_1 \norm{\vv_1}^2 - d_0)}^2} & \tfrac{\norm{\vv_1}^4}{{(d_1 
\norm{\vv_1}^2 - d_0)}^2} & 0 & \ldots & 0 \\
-\tfrac{\norm{\vv_2}^2}{{(d_2 \norm{\vv_2}^2 - d_0)}^2} & 0 & \tfrac{\norm{\vv_2}^4}{{(d_2 
\norm{\vv_2}^2 - d_0)}^2} & \ldots & 0 \\
\vdots & \vdots & \vdots & 
\cdots & \vdots \\
-\tfrac{\norm{\vv_l}^2}{{(d_l \norm{\vv_l}^2 - d_0)}^2} & 0 & 0 & \ldots & 
\tfrac{\norm{\vv_l}^4}{{(d_l \norm{\vv_l}^2 - d_0)}^2}
\end{pmatrix} \succ0
 $$

\noindent leads to strict convexity of the problem. Since  $f_0(\dvv)$ is also coercive, we can 
summarize the existence conditions of optimal solutions in the final proposition as we did for 
NN-(M)DOOLSE problems.

\begin{proposition}[Existence of equivalent (RE)MLE]\label{pro:cvxREMLEexist}
	\leavevmode\\	
	Let assume  $\epsv \not\in \LS{\V}$. Then in a Gaussian orthogonal FDSLRM the following holds
	\begin{enumeratep}
		\item Equivalent \emph{(RE)MLE} problems are strictly convex optimization problems.
		\item The objective function $f_0(\dvv)$ is coercive with respect to constraints.
		\item Their solutions $\hat{\dvv}$ always exist and they are unique global minimizers.
	\end{enumeratep}
\end{proposition}

\begin{proof}[Theorem \ref{thm:eqDOOLSEMLE}] \\
The bijection \eqref{def:bijection} implies $d_0-d_j \norm{\vv_j}^2 = 
(\nus_0+\nus_j\norm{\vv_j}^2)^{-1} \neq 0$ which dictates $\mu_j = 0$ to satisfy the complementary 
slackness conditions in KKT (3)  and simultaneously forms the new version of KKT (1-3) of 
Proposition \ref{pro:KKTcvxREMLE}
\begin{align} 
||\epsv||^2-(n^{*}-l)\nus_0-\Sum{j}{l}(\nus_0+\nus_j\norm{\vv_j}^2)& =0,\label{eq:newKKT1} \\
||\vv_j||^2(\nus_0+\nus_j\norm{\vv_j}^2)-(\epsv'\vv_j)^2-\lambda_j & =0,\,\,j\in\{1,\ldots,l\} 
\label{eq:newKKT2}\\[6pt]
 \nus_0>0, \, -\nus_j \lambda_j = 0, \,  \nus_j \geq 0,  \, \lambda_j & \geq  0, 
 \,\,j\in\{1,\ldots,l\} \label{eq:newKKT3}
\end{align}
But for $\epsv \not\in \LS{\V}$, which occurs with probability 1, the system 
\eqref{eq:newKKT1}-\eqref{eq:newKKT3} is the same as the system \eqref{eq:KKT0}-\eqref{eq:slack},
which is what we set out to prove.\end{proof}

\bibliographystyle{spmpsci}      
\bibliography{references}   

%
%

\end{document}

%% file: OwnSymbolsCommands.tex
\usepackage{graphicx}
\usepackage{amsmath}
\usepackage{amsfonts}
\usepackage{amssymb}
\usepackage{epstopdf}
\usepackage[utf8]{inputenc}
\usepackage[english]{babel}
\usepackage{upgreek}
\usepackage{dsfont}
\usepackage{mathrsfs} 
\usepackage{ctable}
\usepackage{float}	
\usepackage{mathbbol}
\usepackage[justification=centering]{caption}
\usepackage{enumitem}
\usepackage{array,multirow}
\usepackage{caption}
\usepackage{hyperref}
\hypersetup{
	colorlinks=true,
	linkcolor=blue,
	filecolor=magenta,      
	urlcolor=blue,
	citecolor=red,
}
\newcommand{\Section}[1]{\section{#1} \setcounter{equation}{0}}

\newcommand{\mathbi}[1]{\boldsymbol{#1}} 

\newcommand{\eps}{\varepsilon}     

\newcommand{\vertspace}{\vspace{6pt}}

\newcommand{\bv}{\mathbi{\beta}}
\newcommand{\estbv}{\mathbi{\beta^*}}
\newcommand{\Yv}{\mathbi{Y}}
\newcommand{\Xv}{\mathbi{X}}
\newcommand{\xv}{\mathbi{x}}
\newcommand{\wv}{\mathbi{w}}

\newcommand{\vv}{\mathbi{v}}
\newcommand{\jv}[1]{\mathbi{j}_{#1}}

\newcommand{\K}{\mathrm{K}}
\newcommand{\Kb}{\mathrm{K}(\bvv)}
\newcommand{\qv}{\mathbi{q}}
\newcommand{\bvv}{\mathbi{b}}

\newcommand{\GV}{\G_\V}
\newcommand{\invDb}{\D_{\bvv}^{-1}}
\newcommand{\Db}{\D_{\bvv}}

\newcommand{\E}[1]{E\left\{#1\right\}}
\newcommand{\En}[1]{E_{\nuv}\left\{#1\right\}}

\newcommand{\biasn}[1]{\Updelta_{\nuv}\left\{#1\right\}}

\renewcommand{\S}{\mathrm{S}}
\newcommand{\Seps}{\S_{\epsv}}

\newcommand{\epsv}{\boldsymbol{e}} 

\newcommand{\lambdav}{\boldsymbol{\uplambda}}
\newcommand{\lambdas}{\uplambda}
\newcommand{\muv}{\boldsymbol{\mu}}

\newcommand{\nuv}{\boldsymbol{\upnu}} 

\newcommand{\nus}{\upnu}
\newcommand{\gv}{\boldsymbol{\upgamma}}
\newcommand{\gs}{\upgamma}

\newcommand{\s}{\sigma}                         

\newcommand{\Sgn}{\mathrm{\Sigma_{\nuv}}}
\newcommand{\InvSgn}{\mathrm{\Sigma_{\nuv}^{-1}}}
\newcommand{\Cov}[1]{{Cov}\{#1\}}
\newcommand{\Covn}[1]{{Cov}_{\nuv}\left\{#1\right\}}

\newcommand{\Di}[1]{D\left\{#1\right\}}
\newcommand{\Din}[1]{D_{\nuv}\{#1\}}

\newcommand{\MSEn}[1]{{MSE}_{\nuv}\{#1\}}

\newcommand{\V}{\mathrm{V}}
\newcommand{\F}{\mathrm{F}}
\newcommand{\W}{\mathrm{W}}

\newcommand{\lm}{l_{\small \mathrm{m}}}
\newcommand{\lr}{l_{\small \mathrm{r}}}

\newcommand{\I}[1]{\mathrm{I}_{#1}}   

\newcommand{\pitem}{$(\arabic*)$}

\newenvironment{enumeratep}
{\begin{enumerate}[label=\pitem]}
	{\end{enumerate}}

\newcommand{\est}[1]{\tilde{#1}}     

\newcommand{\dest}[2]{{\tilde{#1}}_{#2}}

\newcommand{\nndest}[1]{\tilde{#1}_{+}}

\newcommand{\nesX}[1]{\breve{\s}_{#1}^2(\Xv)}

\newcommand{\best}[1]{\mathring{#1}}  
\newcommand{\bess}[1]{\mathring{{\s}}_{#1}^2}  
\newcommand{\besX}[1]{{\mathring{{\s}}_{#1}}^2(\Xv)} 
\newcommand{\besnX}[1]{{\mathring{{\s}}_{#1}}^2} 

\newcommand{\NnX}{\mathcal{N}_n(\F\bv, \Sgn)}
\newcommand{\NX}{\mathcal{N}}

\newcommand{\bluebn}{\bv_{\nuv}^{*}}
\newcommand{\rwn}{\wv_{\nuv}^{*}}

\newcommand{\blupY}{\Yv^{*}}
\newcommand{\blupZ}{\mathbi{Z}^{*}}
\newcommand{\blupZn}{\mathbi{Z}_{\nuv}^{*}}
\newcommand{\blupYn}{\Yv_{\!\!\nuv}^{*}}

\newcommand{\Ts}{\mathds{T}}
\newcommand{\N}{\mathds{N}}
\newcommand{\Nz}{\mathds{N}_0}
\newcommand{\Rv}[1]{\mathds{R}^{#1}}

\newcommand{\Ys}{\Upupsilon}

\newcommand{\bigO}{\mathcal{O}}

\newcommand{\Rm}[2]{\mathds{R}^{#1 \times #2}}

\newcommand{\LS}[1]{\mathscr{L}(#1)}

\newcommand{\Mf}{\mathrm{M_{F}}}

\newcommand{\Mv}{\mathrm{M_{V}}}

\newcommand{\A}{\mathrm{A}}
\newcommand{\B}{\mathrm{B}}

\newcommand{\D}{\mathrm{D}}

\newcommand{\Hs}{\mathrm{H}}

\newcommand{\dT}{\mathbb{T}}
\newcommand{\dt}{\mathbb{t}}
\newcommand{\dk}{\rho}

\newcommand{\In}{\mathrm{I}_n}
\newcommand{\Il}{\mathrm{I}_l}
\newcommand{\Ov}[1]{\boldsymbol{0}_{#1}}
\newcommand{\Om}[2]{\boldsymbol{0}_{#1 \times #2}}

\newcommand{\G}{\mathrm{G}}

\newcommand{\dvv}{\mathbi{d}}
\newcommand{\dv}[1]{\mathbb{#1}}
\newcommand{\dGV}{\dv{G}_{\V}}
\newcommand{\dHV}{\dv{H}_{\V}}
\newcommand{\dU}{\dv{U}}
\newcommand{\dW}{\dv{W}}

\newcommand{\blockmatrix}[4]{
	\left(\begin{matrix}
		#1 & \; #2 \\
		#3 & \; #4 \\
	\end{matrix}
	\right)
}
\newcommand{\smallblockmatrix}[4]{
	\left(\begin{smallmatrix}
		#1 & #2 \\
		#3 & #4 \\
	\end{smallmatrix}
	\right)
}
\newcommand{\norm}[1]{\left\Vert#1\right\Vert}
\newcommand{\gnorm}[2]{\norm{#1}_{#2}}

\newcommand{\tr}{\mathrm{tr}}

\newcommand{\diag}[1]{\mathrm{diag}\left\{#1\right\}}

\newcommand{\Sum}[2]{\sum\limits_{#1=1}^{#2}}
\newcommand{\SumZ}[2]{\sum\limits_{#1=0}^{#2}}

\newcommand{\pref}[1]{p.~\!\pageref{#1}}